\documentclass[11pt,letterpaper]{article}
\usepackage{style}
\usepackage{shortcuts}

\title{Faster Sublinear-Time Edit Distance}
\date{}

\author{%
    Karl Bringmann\footnote{Saarland University and Max Planck Institute for Informatics, Saarland Informatics Campus. This work is part of the project TIPEA that has received funding from the European Research Council (ERC) under the European Unions Horizon 2020 research and innovation programme (grant agreement No.~850979).}\and%
    Alejandro Cassis\footnote{Saarland University and Max Planck Institute for Informatics, Saarland Informatics Campus. Supported by TIPEA as above.}\and
    Nick Fischer\footnote{Weizmann Institute of Science. This work is part of the project CONJEXITY that has received funding from the European Research Council (ERC) under the European Union's Horizon Europe research and innovation programme (grant agreement No.~101078482).} \and
    Tomasz Kociumaka\footnote{Max Planck Institute for Informatics.}}

\begin{document}
\maketitle

\begin{abstract}
\noindent
We study the fundamental problem of approximating the edit distance of two strings. After an extensive line of research led to the development of a constant-factor approximation algorithm in almost-linear time, recent years have witnessed a notable shift in focus towards \emph{sublinear-time} algorithms. Here, the task is typically formalized as the \emph{$(k, K)$-gap edit distance} problem: Distinguish whether the edit distance of two strings is at most $k$ or more than $K$.

Surprisingly, it is still possible to compute meaningful approximations in this challenging regime.
Nevertheless, in almost all previous work, truly sublinear running time of $\Oh(n^{1-\varepsilon})$ (for a constant $\varepsilon > 0$) comes at the price of at least \emph{polynomial} gap $K \ge k \cdot n^{\Omega(\varepsilon)}$.
Only recently, [Bringmann, Cassis, Fischer, and Nakos; STOC~'22] broke through this barrier and solved the \emph{sub-polynomial} $(k, k^{1+o(1)})$-gap edit distance problem in time $\Order(n/k + k^{4+\order(1)})$, which is truly sublinear if \smash{$n^{\Omega(1)} \le k \le n^{\frac14-\Omega(1)}$}.
The $n/k$ term is inevitable (already for Hamming distance), but it remains an important task to optimize the $\poly(k)$ term and, in general, solve the $(k, k^{1+o(1)})$-gap edit distance problem in sublinear-time for larger values of~$k$.

In this work, we design an improved algorithm for the $(k, k^{1+\order(1)})$-gap edit distance problem in sublinear time $\Order(n/k + k^{2+\order(1)})$, yielding a significant quadratic speed-up over the previous \makebox{$\Order(n/k + k^{4+\order(1)})$}-time algorithm. Notably, our algorithm is unconditionally almost-optimal (up to subpolynomial factors) in the regime where \smash{$k \leq n^{\frac13}$} and improves upon the state of the art for~\smash{$k \leq n^{\frac12-\order(1)}$}. Similarly to previous results, our algorithm is based on the framework of [Andoni, Krauthgamer, and Onak; FOCS~'10], and thus we can further reduce the gap to polylogarithmic ($K = k\cdot (\log k)^{\Order(1/\varepsilon)}$) at the cost of increasing our running time by a factor $k^\varepsilon$.
\end{abstract}

\thispagestyle{empty}
\setcounter{page}{0}
\clearpage

\section{Introduction}
Comparing texts is an essential primitive in computer science, with countless applications in document processing, speech recognition, computational biology and many more domains. In these areas, one of the most popular and well-studied (dis)similarity measures of two strings $X$ and $Y$ is the \emph{edit distance}~$\ED(X, Y)$ (also known as Levenshtein distance~\cite{Levenshtein66}), which is defined as the minimum number of character insertions, deletions, and substitutions to transform $X$ into $Y$. A famous textbook dynamic programming algorithm computes the edit distance of two length-$n$ strings in time~$\Order(n^2)$~\cite{Vintsyuk68,NeedlemanW70,WagnerF74,Sellers74}, and, despite considerable effort, this running time could only be improved by log-factors~\cite{MasekP80,Grabowski16}. More than 50 years after the initial efforts, this quadratic barrier can be explained using the modern toolkit from fine-grained complexity theory~\cite{BackursI18,AbboudBW15,BringmannK15,AbboudHWW16}. At the time though, this barrier was bypassed by the elegant \emph{Landau--Vishkin} algorithm~\cite{LandauV88,LandauMS98}: It computes the edit distance of two strings in time $\Order(n + k^2)$, where the running time depends on the actual edit distance $k = \ED(X, Y)$. There is little hope to optimize the $\Order(n + k^2)$ time (beyond lower-order factors): Even restricted to instances with $k=\Theta(n^\kappa)$ for some $\kappa \in (\frac12, 1]$, a hypothetical $\Order(n+k^{2-\epsilon})$-time algorithm would directly yield an $\Order(n^{2-\epsilon})$-time algorithm for arbitrary edit distances, violating the fine-grained lower bound. Moreover, for $\kappa \in [0,\frac12]$, the dominating $\Oh(n)$ term is necessary simply to read the input strings.

In this paper, we propose a new algorithm that can be viewed as an alternative to the Landau-Vishkin algorithm, with two significant changes: Instead of computing the edit distance exactly, our algorithm settles for an \emph{approximation} of decent quality. In exchange, it avoids reading the entire input strings and runs in \emph{sublinear} time $\widehat\Order(n/k + k^2)$.\footnote{For readability in this introduction, we write $\widehat\Order(\cdot), \widehat\Omega(\cdot)$ and $\widehat\Theta(\cdot)$ to hide subpolynomial factors $n^{o(1)}$. We also write $\widetilde\Order(\cdot), \widetilde\Omega(\cdot)$ and $\widetilde\Theta(\cdot)$ to hide polylogarithmic factors $(\log n)^{O(1)}$. We do not use this notation in the technical parts to avoid ambiguities.}

\subsection{Previous Work}
Computing the edit distance has been a driving question in string algorithms for decades. Even before the aforementioned hardness results for exact edit distance were known, the community spent extensive efforts towards computing accurate \emph{approximations} of the edit distance in subquadratic time. The first major milestones of this long line of research include various polynomial-factor approximations~\cite{LandauMS98,BarYossefJKK04,BatuES06}, a subpolynomial $n^{\order(1)}$-factor approximation in almost-linear time $n^{1+\order(1)}$ by Andoni and Onak~\cite{AndoniO12} (based on the Ostrovsky-Rabani embedding~\cite{OstrovskyR07}), and a polylogarithmic $(\log n)^{\Order(1/\varepsilon)}$-factor approximation in almost-linear time $\Order(n^{1+\varepsilon})$ by Andoni, Krauthgamer and Onak~\cite{AndoniKO10}. This polylogarithmic approximation remained the state of the art for several years until, only recently, Chakraborty, Das, Goldenberg, Kouck{\'{y}} and Saks~\cite{ChakrabortyDGKS20} achieved a breakthrough. Inspired by~\cite{BoroujeniEGHS21}, they gave the first \emph{constant}-factor approximation in subquadratic time~\smash{$\widetilde\Order(n^{12/7})$}. Their work was later extended in two incomparable directions: On the one hand, subsuming~\cite{BrakensiekR20,KouckyS20}, Andoni and Nosatzki~\cite{AndoniN20} improved the running time and developed a constant-factor approximation in time $\Order(n^{1+\varepsilon})$ for any $\varepsilon > 0$. On the other hand, fine-tuning the approximation quality lead to $(3 + \varepsilon)$-factor approximations (for any constant~\makebox{$\varepsilon > 0$}) in truly subquadratic time $n^{1.6+o(1)}$~\cite{ChakrabortyDGKS20,GoldenbergRS20}.

These results form an impressive state of the art for edit distance approximations, and it might seem that a constant-factor approximation in almost-linear time is close to the best that could be hoped for. However, it is, in principle, possible to expect \emph{sublinear-time} approximation algorithms. In this regime, it is standard to study approximations in the guise of the $(k, K)$-\emph{gap edit distance} problem, where the goal is to distinguish whether two strings have edit distance at most $k$ or more than~$K$. An algorithm for this gap problem naturally extends to an approximation algorithm with multiplicative error $K / k$ (the ``gap''). The natural question is whether we can match the linear-time state of the art by sublinear-time algorithms. Specifically:
\begin{center}
    \emph{Question 1: Is $(k, k^{1+\order(1)})$-gap edit distance in truly sublinear time?}\\\emph{What about $(k, \Order(k))$-gap edit distance?}
\end{center}
Of course, for small $k$, we cannot expect any sublinear-time improvements (as already distinguishing the all-zeros string from a string with a single one unconditionally requires reading $\Omega(n)$ characters), so by ``truly sublinear'' we mean running in $n^{1-\Omega(1)}$ time for $k \geq n^{\Omega(1)}$.

Fueled by this driving question, a line of research developed progressively better sublinear-time algorithms; see also \cref{tab:comparison} for the following list of relevant previous work. The first result in this direction is due to Batu, Erg{\"{u}}n, Kilian, Magen, Raskhodnikova, Rubinfeld and Sami~\cite{BatuEKMRRS03} who solved the $(k, \Theta(n))$-gap edit distance problem in sublinear time \smash{$\widetilde\Order(k^2 / n + \sqrt k)$} (subject to the restriction that~\makebox{$k < n^{1-\varepsilon}$} for some constant $\varepsilon > 0$). The Andoni--Onak $n^{\order(1)}$\=/approximation algorithm can further be turned into an algorithm for the $(k, K)$-gap edit distance problem in time~\smash{$\widehat\Order(n^2 k / K^2)$} (provided that $K$ is polynomially larger than $k$)~\cite{AndoniO12}. More recently, based on the Landau--Vishkin algorithm, Goldenberg, Krauthgamer and Saha~\cite{GoldenbergKS19} and Kociumaka and Saha~\cite{KociumakaS20} solved the $(k, \Theta(k^2))$-gap problem in time \smash{$\widetilde\Order(n/k + k^2)$}.  Their algorithm allows for various other trade-offs between gap and running time (see \cref{tab:comparison}).
In a combined effort, they later developed a different algorithm that runs in $\Ohtilde(n\sqrt{k}/K + nk^2/K^2) \subseteq \Ohtilde(nk/K)$ time for $K\ge k^{1+\Omega(1)}$.
This solution is \emph{non-adaptive} (i.e., the queried positions in the string can be fixed in advance). Moreover, they provided a strong barrier and proved that any \emph{non-adaptive} algorithm for the~\makebox{$(k, K)$}-gap problem requires $\Omega(n\sqrt{k}/K)$ queries. In that sense, their algorithm is optimal
for $K\ge k^{3/2}$.

\begin{table}[t]
    \caption{A comparison of sublinear-time algorithms for the $(k, K)$-gap edit distance problem for different gap parameters $k$ and $K$. All algorithms in this table are randomized and succeed with high probability.}\label{tab:comparison}
    \small%
    \setlength{\extrarowheight}{.8ex}%
    \begin{tabular*}{\linewidth}[t]{@{\extracolsep{\fill}}>{\raggedright}p{.52\linewidth}>{\raggedright}p{.21\linewidth}p{.21\linewidth}<{\raggedright}}
        \toprule
        \thead[l]{Source} & \thead[l]{Running time} & \thead[l]{Assumptions} \\
        \midrule
        Batu, Erg{\"{u}}n, Kilian, Magen, Raskhodnikova, Rubinfeld, Sami~\cite{BatuEKMRRS03} & $\widetilde\Order(k^2 / n + \sqrt k)$ & ($k = n^{1-\Omega(1)}$ and~$K = \Omega(n)$) \\
        Andoni, Onak~\cite{AndoniO12} & $\widehat\Order(n^2 k / K^2)$ & ($K > k^{1+\Omega(1)}$) \\
        Goldenberg, Krauthgamer, Saha~\cite{GoldenbergKS19} & $\widetilde\Order(nk/K + k^3)$ & ($K > k^{1+\Omega(1)}$) \\
        Kociumaka, Saha~\cite{KociumakaS20} & $\widetilde\Order((nk+\sqrt{nk^5})/K+ k^2)$ & ($K > k^{1+\Omega(1)}$) \\
        Brakensiek, Charikar, Rubinstein~\cite{BrakensiekCR20} & $\widetilde\Order((n+k^2)\cdot k^{3/2}/K)$ & ($K \ge k^{3/2}$) \\
        Goldenberg, Kociumaka, Krauthgamer, Saha~\cite{GoldenbergKKS22} & $\widetilde\Order(n \sqrt k / K)$ & ($K \ge k^{3/2}$) \\
        Goldenberg, Kociumaka, Krauthgamer, Saha~\cite{GoldenbergKKS22} & $\widetilde\Order(n k^2 / K^2)$ & ($k^{3/2} > K > k^{1+\Omega(1)}$) \\
        Bringmann, Cassis, Fischer, Nakos~\cite{BringmannCFN22} & $\widehat\Order(n/k + k^4)$ & ($K = \widehat\Theta(k)$) \\
        \emph{This work (\cref{cor:subpoly})} & $\widehat\Order(n/k + k^2)$ & ($K = \widehat\Theta(k)$) \\
        \emph{This work (\cref{cor:poly})} & $\widehat\Order(n/K + \sqrt{nk} + k^2)$ & ($K > k^{1+\Omega(1)}$) \\
        \bottomrule
    \end{tabular*}
\end{table}

At the same time, finally, Bringmann, Cassis, Fischer and Nakos~\cite{BringmannCFN22} resolved Question~1 and developed the first truly sublinear algorithm with a \emph{subpolynomial} gap. For simplicity, from now on we will refer to their algorithm as the \emph{BCFN algorithm.} The BCFN algorithm solves the $(k, k^{1+\order(1)})$-gap edit distance problem in time $\Order(n / k + k^{4+\order(1)})$, which is indeed truly sublinear for the range of parameters \smash{$n^\varepsilon \leq k \leq n^{\frac14-\varepsilon}$}. The algorithm is adaptive, as is necessary by the previously mentioned lower bound~\cite{GoldenbergKKS22}. Moreover, it is based on the Andoni--Krauthgamer--Onak algorithm~\cite{AndoniKO10} and comes with the benefit that the approximation factor can even be polylogarithmic ($K = k (\log k)^{\Order(1/\varepsilon)}$) if we increase the running time by a factor $k^\varepsilon$.

In summary, the BCFN algorithm provides a satisfying and, perhaps, surprising answer to the question of whether we can expect accurate approximations in truly sublinear time. Nevertheless, the range of parameters for which it becomes effective is quite limited. For instance, for $k \geq n^{\frac14-\order(1)}$ the BCFN algorithm is even outperformed by the classic Landau--Vishkin algorithm that computes the edit distance \emph{exactly}. In this paper, we therefore press on and study the equally important follow-up question:
\begin{center}
    \medskip
    \emph{Question 2: What is the time complexity of the $(k, k^{1+\order(1)})$-gap edit distance problem?}
    \medskip
\end{center}
Let us remark that the strongest known unconditional lower bound is~\smash{$\widehat\Omega(n/k + \sqrt n)$}~\cite{BatuEKMRRS03,AndoniN10}, where the $n/k$ term is necessary already for the gap \emph{Hamming} distance problem. In light of this lower bound, the BCFN running time $\widehat\Order(n/k + \poly(k))$ has the right format, except that it remains to optimize the $\poly(k)$ term.

\subsection{Our Results}
Our contribution is that we drastically reduce the $\poly(k)$ term, thereby making significant progress towards our driving Question~2. Specifically, our main result is the following technical theorem that we will shortly instantiate for several interesting parameter settings:

\begin{restatable}[Main Theorem]{theorem}{thmmain} \label{thm:main}
    Let $2 \leq \Delta \leq n$ be a parameter.
    Then, there is a randomized algorithm that solves the $(k, K)$-gap edit distance problem in time $\Order(n/K + kK) \cdot \Delta^{3} \cdot (\log n)^{O(\log_\Delta n)}$ and succeeds with constant probability,
    provided that $K/k \geq (\log n)^{c \cdot \log_\Delta(n)}$ for a sufficiently large constant $c > 0$.
\end{restatable}

The main consequence is a significant improvement for the $(k, k^{1+\order(1)})$-gap edit distance problem:

\begin{restatable}[Subpolynomial Gap]{corollary}{corsubpoly} \label{cor:subpoly}
The $(k, k \cdot 2^{\Theta(\sqrt{\log k \log\log k})})$-gap edit distance problem is in time $\Order(n/k + k^{2+\order(1)})$.
\end{restatable}

Our algorithm constitutes a \emph{quadratic} improvement over the BCFN algorithm~\cite{BringmannCFN22}. As a consequence, we extend the range of parameters for which we know truly sublinear-time algorithms for the $(k, k^{1+\order(1)})$-gap problem to $n^\varepsilon \leq k \leq n^{\frac12-\varepsilon}$. In a slightly narrower range of  $n^\varepsilon \leq k \leq n^{\frac13-\varepsilon}$, our algorithm matches the unconditional lower bound~\cite{BatuEKMRRS03,AndoniN10} and is therefore \emph{almost-optimal}, up to subpolynomial factors in the gap and in the running time.

Another notable consequence is that our running time is never out-performed by the Landau-Vishkin algorithm. In other words, as mentioned before, our algorithm can even be viewed as an alternative to the $\Order(n + k^2)$-time Landau--Vishkin algorithm: It runs in faster sublinear time $\widehat\Order(n/k + k^2)$ at the cost of returning an approximate result.

Similarly to the BCFN algorithm, we also build on the original framework of Andoni, Krauthgamer and Onak~\cite{AndoniKO10}. As a consequence, our Main~\cref{thm:main} allows for more precise approximations with \emph{polylogarithmic} gap, at the mild cost of increasing the running time by a small polynomial factor:

\begin{restatable}[Polylogarithmic Gap]{corollary}{corpolylog} \label{cor:polylog}
For any constant $\varepsilon > 0$, the $(k, k \cdot (\log k)^{\Theta(1/\varepsilon)})$-gap edit distance problem is in time \smash{$\Order(n/k^{1-\varepsilon} + k^{2 + \varepsilon})$}.
\end{restatable}

Finally, we also give new results in the regime where $K$ is polynomially larger than $k$.

\begin{restatable}[Polynomial Gap]{corollary}{corpoly}\label{cor:poly}
Let $k, K$ be such that $K > k^{1+\varepsilon}$ for some constant $\varepsilon > 0$.
Then the $(k, K)$-gap edit distance problem is in time $\Order((n/K + \sqrt{nk} + k^2) \cdot n^{o(1)})$.
\end{restatable}

This polynomial gap regime was not explicitly considered by \cite{BringmannCFN22}.
However, using their results one can infer that they obtain an algorithm for this task in time $\widehat{O}(n/K + n^{0.8} + k^4)$.
Our result improves upon this in all parameter regimes.

\subsection{Technical Overview}
In this section, we give a high-level overview of our new ideas.
The starting point for our result is the BCFN algorithm.
It is based on two main ingredients: (1) The Andoni--Krauthgamer--Onak framework providing an efficient recursion scheme, based on the so-called \emph{precision sampling} technique, and (2) a \emph{structure-versus-randomness} dichotomy that is used to bound the number of recursive subproblems by $\poly(k)$.
Since the algorithm is quite complex, and since our improvements are only concerned with the second ingredient, we omit an extensive description of the BCFN algorithm here; we refer the interested reader to the original paper for a thorough overview (or to our technical \cref{sec:bcfn}).\footnote{In particular, throughout this overview, we will ignore the term ``$n/k$'' in all running times as dealing with these terms requires precision sampling which is not our focus here. More precisely, while we in this overview implicitly pretend that in each subproblem we query the strings at $\widehat\Order(n/k + \poly(k))$ positions, we in fact query at $\widehat\Order(pn + \poly(k))$ positions where $p$ is a precision parameter distributed with $\Ex[p] \leq \widehat\Order(1/k)$.}
Instead, we focus on highlighting our new key concepts.

One of the central ideas behind the BCFN algorithm is the notion of \emph{block periodicity}, formally defined as follows:

\begin{definition}[Block Periodicity]\label{def:block-periodicity}
Let $X$ be a string and $p \geq 1$ be an integer. The \emph{$p$-block periodicity $\BP_p(X)$} of $X$ is the smallest integer $L$ such that $X$ can be partitioned into $L$ substrings, $X = \bigodot_{\ell=1}^L X_\ell$, so that each substring $X_\ell$ is $p$-periodic, that is, the period of $X_\ell$ is at most $p$.
\end{definition}

The block periodicity is a natural measure of \emph{structure} in a string: On the one hand, strings with minimal block periodicity ($\BP_p(X) = 1$) are simply \emph{periodic}. On the other hand, strings with very large block periodicity are composed of many non-repetitive parts and often behave like \emph{random} strings. A frequent phenomenon in string algorithms is that specialized techniques for these two extreme cases ultimately lead to algorithms for the entire spectrum.

An important stepping stone towards the general BCFN algorithm for the $(k, k^{1+\order(1)})$-gap edit distance problem is the special case where $X$ has bounded block periodicity $\BP_p(X) \leq B$ whereas~$Y$ can be arbitrary.
Let us refer to this special case as the \emph{block-periodic case}. In our setting, $p$ and~$B$ are parameters subject to some technical conditions (such as $p, B \geq k$) that will not concern us here.
The BCFN algorithm can be interpreted in the following way:
\begin{enumerate}
    \item An efficient algorithm for the block-periodic case. Specifically, it follows from~\cite{BringmannCFN22} that the block-periodic case can be solved in time \smash{$\widehat\Order(n/k + Bpk)$}.
    \item A reduction from the general to the block-periodic case. This reduction picks as parameters $p, B = \widehat\Order(k)$ and requires an additional multiplicative overhead of $k$. In combination with the previous item, the running time thus becomes $\widehat\Order(n/k + k^4)$.
\end{enumerate}
We remark that this interpretation is not immediate from the original paper~\cite{BringmannCFN22} and requires rearranging the algorithm up to some degree.\footnote{In particular, in~\cite{BringmannCFN22} the authors did not explicitly study the block-periodic case as a subproblem, nor did they consider $p$ and $B$ as parameters on their own.}
In this paper, we achieve our results by separately improving both of these steps.

\paragraph{Improvement 1: Speed-Up for the Block-Periodic Case}
As our first major contribution, we achieve a factor-$k$ speed-up for the block-periodic case; see the following simplified statement:

\begin{lemma}[Simplified Version of \cref{lem:main-stoc}]\label{lem:informal-stoc}
Let $X$ be a string with block periodicity $\BP_p(X) \leq B$.
The $(k, k^{1+o(1)})$-gap edit distance problem for $X$ and an arbitrary string $Y$ is in time \smash{$\widehat\Order(n/k + B p)$}.
\end{lemma}

Our strategy refines the approach by~\cite{BringmannCFN22} as follows.
Intuitively, the algorithm keeps splitting the string $X$ into smaller substrings $X_1, \dots, X_N$ until all of them are $p$-periodic. Since the block periodicity of $X$ was initially bounded by $B$, we construct at most \smash{$N = \widehat\Order(B)$} many substrings in this way.
Moreover, for almost all such substrings $X_i$, we can assume that they perfectly match their respective substrings $Y_i$ up to some shift---if there are more than $k$ exceptions, we can immediately infer that the edit distance between $X$ and $Y$ exceeds $k$ and terminate the algorithm.
So far, there is no difference to the original BCFN algorithm. Our improvement lies in the treatment of the periodic pieces $(X_i, Y_i)$.
Specifically, we show the following lemma.

\begin{lemma}[Simplified Version of \cref{lem:approxed-periodic-PM}] \label{lem:informal-periodic}
Let $X, Y$ be strings of lengths $n$ and $n + k$, respectively. Given an integer $p$ that is a period of both strings, we can compute $\widehat\Order(1)$-approximations of the edit distances $\ED(X, Y \intervalco{i}{n+i})$ for all $i \in \intervalcc{0}{k}$ in overall time $\widehat\Order(p + k)$.
\end{lemma}

In comparison, the analogous result in~\cite{BringmannCFN22} computes the edit distances exactly (capped with $\Oh(k)$),
but the running time increases by a multiplicative factor $k$.
The reason why we have to compute the edit distance for \emph{many shifts} is due to the intricate recursion scheme of the Andoni--Krauthgamer--Onak framework, which is applied throughout under the hood.

\cref{lem:informal-periodic} is proven by combining two ideas:
First, using structural insights on the edit distance of periodic strings, we can reduce the problem to strings of length $\Oh(p)$ (see \cref{lem:ed-approx-periodic}; here, we incur a factor-3 loss in the approximation ratio).
Second, we apply in a black-box fashion a known algorithm that computes the edit distance between two strings for many shifts.
The Andoni--Onak algorithm~\cite{AndoniO12} is suitable for this task but only attains a subpolynomial approximation.
To improve our approximation factor to polylogarithmic, we instead use the very recent \emph{dynamic} edit distance approximation algorithm by Kociumaka, Mukherjee, and Saha~\cite{KMS23}; see \cref{sec:shifts} for more details.

In summary, these ideas are sufficient to improve the BCFN algorithm to run in time $\widehat\Order(n/k + k^3)$.
The following insights further reduce the running time to $\widehat\Order(n/k + k^2)$.

\paragraph{Improvement 2: Block Periodicity Decomposition}
Our second and technically much more challenging contribution is to improve the reduction from the general case to the block-periodic case. Specifically, building on the $\widehat\Order(n/k + Bp)$-time algorithm for the block-periodic case, our goal is to develop an $\widehat\Order(n/k + k^2)$-time algorithm for the general $(k, k^{1+\order(1)})$-gap edit distance problem.

Our reduction diverges completely from the approach of \cite{BringmannCFN22}.
Instead, we rely on more structural insights related to the block periodicity, especially on the notion of \emph{breaks}.
For a string~$X$, we say that a position $i$ is a \emph{$p$-break} if $i$ is a multiple of~$p$ and $X \intervalco{i}{i+3p}$ is not $p$-periodic (see \cref{def:break}). We exploit breaks in two ways. First, we observe that the number of $p$-breaks approximates the $p$-block periodicity of a string within a constant factor (see \cref{lem:equivalence-blocks-bp}). Since we can test if $X\intervalco{i}{i+3p}$ is a $p$-break in time $O(p)$~\cite{KnuthMP77}, this insight yields an efficient sublinear-time subroutine estimating the block periodicity of a string.
In the BCFN algorithm, the block periodicity of a string was never explicitly computed, and this new ``Block Periodicity Test'' on its own already vastly streamlines the BCFN algorithm.

Second, we use breaks to solve the $(k, k^{1+o(1)})$-gap problem in a more direct way. Suppose that $\ED(X, Y) \leq k$ and fix an optimal edit distance alignment between $X$ and $Y$.
Let us sample indices~$i$ which are multiples of $k$ at rate $\approx 1/k^2$ and identify all $k$-breaks in $X$ among them.
Then, with constant probability, none of the sampled breaks contain edits performed by the optimal alignment.
In particular, we can infer how the optimal alignment matches each of the sampled breaks.
If the alignment matches a break at position $i$, then the break must have an exact occurrence in $Y$ at position $j\in \intervalcc{i-\floor{k/2}}{i+\floor{k/2}}$ (recall that we assume $|X|=|Y|$). Moreover, since the break is not $k$-periodic, it cannot have more than one such exact occurrence.
For each sampled break, we use exact pattern matching~\cite{KnuthMP77} to find the unique match in $Y$.
This allows us to split the instance into substrings $X_1, \dots, X_s$ and $Y_1, \dots, Y_s$ such that $\sum_i \ED(X_i, Y_i) = \ED(X, Y)$. Moreover, by the aforementioned correspondence between breaks and block periodicity, with good probability the $k$-block periodicity of each substring $X_i$ is bounded by~\smash{$\widetilde\Order(k)$}. This splitting procedure forms the heart of our algorithm (see \cref{lem:split}).

This suggests the following algorithm.
If the $k$-block periodicity of $X$ is $\widehat\Order(k)$, then the instance falls into the block-periodic case, and it can be solved directly using \cref{lem:informal-stoc} in time $\widehat\Order(n/k + Bp) = \widehat\Order(n/k + k^2)$.
Otherwise, we apply the splitting routine to partition the strings into substrings $X_1, \dots,X_s$ and $Y_1, \dots, Y_s$ in such a way that $\ED(X, Y) = \sum_i \ED(X_i, Y_i)$ and the block periodicity of each piece $X_i$ is bounded by \smash{$\BP_k(X_i) \leq \widetilde\Order(k)$}. It remains to distinguish whether $\sum_i \ED(X_i, Y_i) \leq k$ or $\sum_i \ED(X_i, Y_i) > K$ (for $K = k^{1+\order(1)}$).

Naively recursing on all subproblems is not efficient enough as there could be too many of them. Instead, a more careful approach is to \emph{subsample} the subproblems and to recurse only on few of them. This is exactly the right task for the \emph{precision sampling} technique, which has already proven to be an instrumental tool for approximate edit distance~\cite{AndoniKO10,BringmannCFN22,BringmannCFN22b,KMS23}.
We remark that, in our algorithm in \cref{sec:main}, we apply the technique in a more elementary way, similar to~\cite{IndykW05}.
This variant incurs an $\Oh(\log K)$-factor loss in the approximation, but we can still use it, even when aiming for a polylogarithmic gap,
because our recursion is relatively shallow compared to~\cite{AndoniKO10}.

\subsection{Open Questions}

We mention two tantalizing open problems:

\begin{enumerate}
    \item While our result makes considerable progress towards answering Question 2 above, it does not settle it.
    In particular, is there an algorithm for the $(k, k^{1+o(1)})$-gap edit distance problem in time $\widehat{O}(n/k + \sqrt{n})$?
    Alternatively, can we prove a stronger (possibly conditional) lower bound?
    \item Can the $(k, O(k))$-gap edit distance problem be solved in time $\widehat{O}(n/k + \poly(k))$?
    We believe that our approach is hopeless to answer this since we build upon the Andoni--Krauthgamer--Onak framework~\cite{AndoniKO10}, which inherently incurs a polylogarithmic factor in the approximation.
\end{enumerate}

\subsection{Outline}
We structure the remainder of this paper as follows. In \cref{sec:preliminaries} we give some preliminaries. In \cref{sec:main} we state our main algorithm (Improvement 2 from the overview) and prove our Main \cref{thm:main} and its corollaries.
Finally, in \cref{sec:bcfn,sec:shifts} we state the improved algorithm for the block-periodic case (Improvement 1).

\section{Preliminaries} \label{sec:preliminaries}

For integers $i, j \in \Int$, we write $\intervalco{i}{j} := \set{i, i+1, \dots, j-1}$ and $\intervalcc{i}{j} := \set{i,i+1,\dots,j}$.
The sets $\intervaloc{i}{j}$ and $\intervaloo{i}{j}$ are defined analogously.
We set $\poly(n) = n^{O(1)}$ and $\polylog(n) = (\log n)^{O(1)}$.

\paragraph*{Strings}
A string $X = X\access{0}\dots X\access{n-1} \in \Sigma^n$ is a sequence of $|X| = n$ symbols from an alphabet~$\Sigma$.
We usually denote strings by capital letters $X, Y, Z$.
For integers $i, j$ we denote by $X\intervalco{i}{j}$ the substring with indices in $\intervalco{i}{j}$.
Sometimes we call $X\intervalco{i}{j}$ a \emph{fragment} of $X$.
If the indices are out of bounds, we set $X\intervalco{i}{j} = X\intervalco{\max(0, i)}{\min(j, |X|)}$, and similarly for $X\intervalcc{i}{j}$.

For a string $X$ and an integer $s$, we denote by $X^{\circlearrowright s}$ the cyclical rotation of $X$ defined as $X^{\circlearrowright s}\access{i} = Y\access{(i+s) \bmod |X|}$ for $i \in \intervalco{0}{|X|}$.
We say that $X$ is \emph{primitive} if all the non-trivial rotations of $X$ are distinct from itself.
For a string $Q$, we denote by $Q^*$ the infinite-length string obtained by repeating $Q$.
We say that $X$ is \emph{periodic with period $Q$} if $X = Q^*\intervalco{0}{|X|}$.
For an integer~\makebox{$q \geq 1$}, we say that $X$ is \emph{$q$-periodic} if it is periodic with some period of length at most $q$.
We refer to the smallest period length of $X$ as $\per(X)$.

\paragraph*{Hamming and Edit Distances}
Given two strings $X, Y$ of the same length, we define their Hamming distance $\HD(X, Y) := |\set{i \mid X[i] \neq Y\access{i}}|$ as the number of indices in which they differ.
For two strings~$X, Y$ (with possibly different lengths), we define their \emph{edit distance $\ED(X, Y)$} as the minimum number of insertions, deletions or substitutions necessary to transform $X$ into $Y$.
We refer to insertions, deletions and substitutions as \emph{edits}.

We formally define the gap edit distance problem as follows.

\begin{definition}[Gap Edit Distance]
The $\GapED(k, K)$ problem is to distinguish, given two strings~$X, Y$, whether
\smallskip
\begin{itemize}
    \item $\ED(X, Y) \leq k$ (return \Close{} in this case), or
    \item $\ED(X, Y) > K$ (return \Far{} in this case).
\end{itemize}
\end{definition}

We say that an algorithm correctly solves the $\GapED(k,K)$ problem if it returns the correct answer with constant probability.

\paragraph*{Machine Model}
We work under the standard word RAM model, where the words have size logarithmic in the input size of the problem.
That is, given input strings $X, Y$ of total length $n$ over an alphabet $\Sigma$, we assume the words have size $w = \Theta(\log n + \log|\Sigma|)$.

\section{Faster Algorithm} \label{sec:main}

In this section we prove \cref{thm:main}.
We will make use of the following subroutine which solves the gap edit distance problem for strings $X, Y$ when the block periodicity of $X$ is bounded.
We defer its proof to \cref{sec:bcfn}.

\begin{restatable}[Algorithm for Bounded Block Periodicity]{lemma}{lemmainstoc}\label{lem:main-stoc}
    There exists an algorithm $\alg{AlgSmallBP}$ that, given two strings $X,Y$ of length at most $n$, and parameters $k, K, p, B,\Delta\in \Int_+$ such that \begin{enumerate*}[label=(\roman*)] \item $\BP_p(X) \leq B$, \item $p, B \geq k$, \item $(256 \log K)^2 \leq \Delta \leq n$, and \item $K/k \geq (\log n)^{\beta \cdot \log_\Delta(n)}$ where $\beta > 0$ is a sufficiently large constant,\end{enumerate*}
    solves the $\textsc{GapED}(k, K)$ problem on $X, Y$ with probability at least 0.9 and runs in time
    \[
        O\left(\left(\frac{n}{K} \cdot \Delta + p B \cdot \Delta\right) \cdot (\log n)^{\alpha \cdot \log_{\Delta}(n)}\right)
    \]
    for some constant $\alpha > 0$.
\end{restatable}

The remainder of this section is organized as follows.
In \cref{sec:main:sec:breaks} we introduce the crucial notion of breaks and their usage as a proxy for the block periodicity.
In \cref{sec:main:sec:algorithm} we give our main algorithm, in which we leverage \cref{lem:main-stoc}.
Finally, we put things together in \cref{sec:main:sec:main-thms} to prove our main theorems.

\subsection{Breaks and their Usage}\label{sec:main:sec:breaks}

We start by formally defining breaks.

\begin{definition}[Break]\label{def:break}
    Let $X$ be a string of length $n$ and let $k\ge 2$ be an integer.
    A position $i\in \intervalcc{0}{n-3k}$ is a \emph{$k$-break in $X$} if $i$ is a multiple of $k$
    and $\per(X\intervalco{i}{i+3k})>k$, that is, $X\intervalco{i}{i+3k}$ is \emph{not} $k$-periodic.
\end{definition}

The importance of breaks is that their number is a good approximation of the block periodicity of a string, as captured by the following lemma.

\begin{lemma}\label{lem:equivalence-blocks-bp}
    Let $b$ be the number of $k$-breaks in $X$.
    Then, $\frac13 b \le \BP_k(X) \le b+3$.
\end{lemma}
\begin{proof}
    We first argue that $\BP_k(X) \leq b+3$.
    For this, we partition $X$ into (at most) $b+3$ pieces,
    splitting $X$ at position $i+2k$ for every break $i$, as well as at positions $k$ and $n-k$.
    It suffices to prove that each of the resulting pieces has period at most $k$.
    This property is trivially satisfied for pieces contained within $X\intervalco{0}{k}$
    and $X\intervalco{n-k}{n}$ (their length does not exceed $k$).
    So, consider a piece $X\intervalco{p}{q}$ with $k \le p < q \le n-k$
    and let $r\in \intervalcc{p}{q}$ be the largest position such that $\per(X\intervalco{p}{r}) \le k$.
    For a proof by contradiction, suppose that $r<q$.
    Consider an integer $i\in \intervaloc{r-3k}{r-2k}$ that is a multiple of $k$.
    Observe that $r \ge p+k \ge 2k$ implies $i\ge 0$ and $r < q \le n-k$ implies $i\le n-3k$.
    Moreover, there is no piece starting at position $i+2k\in \intervaloc{r-k}{r} \subseteq \intervaloc{p}{r}$, and thus $i$ is not a break.
    This means that $\per(X\intervalco{i}{i+3k})\le k$.
    If $i \le p$, then $\per(X\intervalcc{p}{i+3k})\le \per(X\intervalco{i}{i+3k}) \le k$, contradicting the choice of $r<i+3k$.
    Otherwise, the intersection $X\intervalco{i}{r}$ of $X\intervalco{i}{i+3k}$ and $X\intervalco{p}{r}$ contains at least $2k$ characters.
    By the periodicity lemma~\cite{FineW65}, since both fragments have periods at most $k$, their union $X\intervalco{p}{i+3k}$
    also has period at most $k$, contradicting the choice of $r$.

    Next, we argue that $b \leq 3\BP_k(X)$.
    Let $L := \BP_k(X)$.
    By definition, we can write $X$ as $X = X_1 X_2 \cdots X_L$, where each $X_i$ is $k$-periodic.
    Observe that every substring inside some $X_i$ is $k$-periodic, so no break can be contained inside any $X_i$.
    In particular, every break contains the starting position of at least some $X_i$.
    Moreover, the starting position of any $X_i$ is contained in at most three breaks (since breaks overlap).
    Therefore, the number of breaks is bounded by $3L$.
\end{proof}

At the heart of our algorithm is the following splitting algorithm.
The idea is to use breaks in order to split an instance into independent subproblems, where each subproblem has smaller block periodicity.

\begin{lemma}[Splitting]\label{lem:split}
    There is an algorithm $\alg{Split}(X,Y,k,K,\delta)$ that, given as input two length-$n$ strings~$X,Y$, two thresholds $1\le k\le K$,
    and a parameter $\delta\in (0,1]$, returns partitions $X=X_1\cdots X_s$ and $Y=Y_1\cdots Y_s$ such that:
    \begin{enumerate}
        \item With probability at least $1-\delta$, the inequality $\BP_k(X_i) \leq 4K$ is satisfied for all $i\in \intervalcc{1}{s}$.
        \item If $\ED(X, Y) \leq k$, then $\sum_{i=1}^s \ED(X_i, Y_i) = \ED(X, Y)$ holds with probability at least $1 - \tfrac{3k}{K}\log\frac{n}{\delta}$.
    \end{enumerate}
    The algorithm runs in expected time $O\left(\frac{n}{K} \log \frac{n}{\delta}\right)$.
\end{lemma}
\begin{proof}
    We sample indices $i \in \intervalcc{0}{n-3k}$ which are multiples of $k$ uniformly at random with rate $\frac{1}{K} \log \frac{n}{\delta}$.
    We identify all $k$-breaks in $X$ among the sampled indices.
    For each $k$-break $i$, we try to find a position $j\in \intervalcc{i -\floor{k/2}}{i+\floor{k/2}}$
    such that $X\intervalco{i}{i+3k}=Y\intervalco{j}{j+3k}$.
    Since $i$ is a break, there is at most one such position $j$.
    If there are none, we set $j:=i$.
    In either case, we split $X$ and $Y$ at positions $i$ and $j$, respectively.
    See \cref{alg:split} for the pseudocode.

    Let us first analyze the running time.
    For each index $i$,  computing $\per(X\intervalco{i}{i+3k})$ and finding the occurrences of $X\intervalco{i}{i+3k}$
    within $Y\intervalco{i-\floor{k/2}}{i+3k+\floor{k/2}}$ takes $O(k)$ time since period finding and pattern matching are in linear time~\cite{KnuthMP77}.
    In expectation, the number of sampled indices $i\in \intervalcc{0}{n-3k}$ that are multiples of $k$
    is $O(\frac{n}{k}\cdot \frac{1}{K} \log\frac{n}{\delta})$.
    Therefore, the algorithm runs in $O(\frac{n}{K}\log\frac{n}{\delta})$ expected time, as desired.

    Next, let us analyze the correctness.
    For property (1), note that it suffices to show that we sample at least one $k$-break out of every $K$ consecutive $k$-breaks.
    Indeed, if this claim holds, then, by \cref{lem:equivalence-blocks-bp},
    we obtain that every phrase $X_i$ satisfies $\BP_k(X_i) \leq K + 3 \leq 4K$.
    We proceed to bound the probability of this event.
    Fix $K$ consecutive breaks.
    The probability that we do not sample any of them is
    \[
      \left(1 - \frac{\log(n/\delta)}{K}\right)^{K}
      \leq \exp(-\log (n/\delta)) = \frac{\delta}{n}.
    \]
    By a union bound over the at most $n$ choices of $K$ consecutive $k$-breaks, we get property (1).

    For property (2), suppose that $\ED(X, Y) \leq k$ and fix an optimal alignment.
    Then, the optimal alignment performs at least one edit in at most $3k$ fragments $X\intervalco{i}{i+3k}$ where $i \in \intervalco{0}{n-3k}$ is a multiple of $k$.
    Thus, the probability that we do not sample any such index $i$ is at least
    \[
        \left(1 - \frac{\log(n/\delta)}{K}\right)^{3k} \geq 1 - \frac{3k}{K}\log\frac{n}{\delta}.
    \]
    Condition on this event.
    This means that the fixed optimal alignment matches the sampled breaks perfectly to fragments of $Y$.
    Moreover, since they are breaks, they have a unique perfect match in $Y$.
    This implies the claim.
\end{proof}

\begin{algorithm}[t]
\caption{} \label{alg:split}
\begin{algorithmic}[1]
\Procedure{Split}{$X, Y, k, K,\delta$}
    \State Sample a set $S \subseteq \intervalcc{0}{n-3k}$ of multiples of $k$
    \Statex[2] including each element independently with probability $\frac{1}{K}\log\frac{n}{\delta}$
    \State Let $i_1 < i_2 < \cdots < i_{s-1}$ be the $k$-breaks in $S$
    \For{$\ell \in \intervalco{1}{s}$}
        \If{there exists $z \in \intervalcc{i_\ell - \floor{k/2}}{i_\ell +\floor{k/2}}$ with $X\intervalco{i_\ell}{i_\ell+3k} = Y\intervalco{z}{z + 3k}$}
            \State $j_\ell := z$
        \Else
            \State $j_\ell := i_\ell$
        \EndIf
    \EndFor
    \State Set $i_0 := j_0 := 0$ and $i_s := |X|, j_s := |Y|$
    \State Set $X_\ell := X\intervalco{i_{\ell-1}}{i_\ell}, Y_{\ell} := \intervalco{j_{\ell-1}}{j_\ell}$ for $\ell \in \intervalcc{1}{s}$
    \State \Return $X_1,Y_1,\dots,X_s,Y_s$
\EndProcedure
\end{algorithmic}
\end{algorithm}

\subsection{Algorithm}\label{sec:main:sec:algorithm}

We will make use of the following folklore tool that approximately tests whether two strings are equal.
We sketch the proof for the benefit of the reader (consult e.g. \cite[Lemma 23]{BringmannCFN22} for a full proof).

\begin{lemma}[Equality Test] \label{lem:equality-test}
    Let $X, Y$ be two strings of the same length, and let $r > 0$ be a sampling rate. There is an algorithm which returns one of the following two outputs:
    \begin{itemize}
    \item \Close, in which case $\HD(X, Y) \leq 1/r$.
    \item \Far, in which case $X \neq Y$.
    \end{itemize}
    The algorithm runs in time $\Order(r |X| \log(\delta^{-1}))$ and is correct with probability at least $1 - \delta$.
\end{lemma}
\begin{proof}[Proof sketch]
    Sample $r|X|\log(1/\delta)$ positions $i \in \intervalco{0}{|X|}$ uniformly at random and test if $X[i] = Y[i]$ for every sampled position.
    If no error is found, return \Close{}; otherwise, return \Far{}.
    If $X = Y$, then the algorithm always returns \Close{}.
    Otherwise, if $\HD(X, Y) > 1/r$, then the algorithm returns \Far{} probability at least $1 - \delta$.
\end{proof}

We are now ready to present our algorithm.
Since this is the most technically involved part of our paper, we start with an informal overview to convey some intuition.

\paragraph*{Algorithm Overview}
The high-level idea of the algorithm is as follows:
If the block periodicity of $X$ is bounded, say $\BP_k(X) \leq K$, then we can solve the $\GapED(k, K)$ problem directly using \cref{lem:main-stoc}.
Otherwise, we apply \cref{lem:split} to split the strings into pieces $X_1,Y_1,\dots,X_s,Y_s$ so that $\ED(X, Y) = \sum_i \ED(X_i, Y_i)$.
It remains to distinguish whether $\ED(X, Y) = \sum_i \ED(X_i, Y_i) < k$, or $\ED(X, Y) = \sum_i \ED(X_i, Y_i) > K$.
\alg{Split} guarantees that, with good probability, the block periodicity of each piece is bounded.
Hence, we naturally want to recurse on some of these pieces.
However, we cannot afford to naively recurse on all them since their total size is too large, and we are aiming for sublinear time.
This is exactly the task where we can apply the precision sampling technique, which enables us to recurse in a few subproblems of total length $\approx n/K$.
To obtain our desired running time, we additionally need that the recursive calls run in time proportional to their edit distance $\ED(X_i, Y_i)$, so that we can \emph{distribute} our time budget among them.

To implement the ideas above, we need to overcome some obstacles.
Importantly, observe that \alg{Split} (see \cref{lem:split}) guarantees that $\sum_i \ED(X_i, Y_i) = \ED(X, Y)$ holds with good probability \emph{only in the case when} $\ED(X, Y) \leq k$.
If $\ED(X, Y) > k$ or \alg{Split} fails, then all we know is that the split satisfies $\sum_i \ED(X_i, Y_i) \geq \ED(X, Y)$ due to subadditivity of edit distance.
Moreover, note that a priori there is no way to know whether $\sum_i \ED(X_i, Y_i) = \ED(X, Y)$ holds.
With this in mind, we design two routines \alg{AlgMain} and \alg{AlgBoosted}, see \cref{alg:main} for the pseudocode.
These routines receive as inputs strings $X, Y$, thresholds $k < K$ and a parameter $p$ with the promise that $\BP_p(X) \leq K/k \cdot p$ (additionally, \alg{AlgBoosted} receives a parameter $\delta$ that determines its failure probability).
Intuitively, \alg{AlgMain} solves the $\GapED(k, K)$ problem and runs in the desired running time\footnote{For this informal overview, we use the notation $\widehat{O}(\cdot)$ to ignore subpolynomial factors $n^{o(1)}$ for the purpose of readability, the actual proof does not hide these factors.} $\widehat{O}(n/K + kK)$ \emph{if we have the promise} that $\ED(X, Y) \leq k$ (i.e. in the \Close{} case).
\alg{AlgBoosted} solves the $\GapED(k, K)$ problem with high probability and runs in time $\widehat{O}(n/K + \min(k,\ED(X, Y)) \cdot K)$.
Crucially, one of the terms in the running time of \alg{AlgBoosted} is proportional to $\ED(X, Y)$ (this allows us to distribute the time budget among subproblems) and unlike \alg{AlgMain}, it does not assume that $\ED(X, Y) \leq k$.

\alg{AlgMain} carries out the ideas described above---namely, if the block periodicity is bounded then we solve the problem directly using \cref{lem:main-stoc} in \cref{alg:main:line:base-case-call} (after handling some trivial base cases in \crefrange{alg:main:line:testlargek}{alg:main:line:base-equality}).
Otherwise, we call \alg{Split} in \cref{alg:main:line:split} and perform precision sampling in \crefrange{alg:main:line:pslstart}{alg:main:line:returnfar}.
In order to distribute the time budget among the sampled subproblems, we make calls to \alg{AlgBoosted} in~\cref{alg:main:line:recursion}.
\alg{AlgBoosted} in turn solves the $\GapED(k, K)$ by calling to \alg{AlgMain}.
In order to run in time $\widehat{O}(n/K + \min(k, \ED(X, Y)) K)$, we use exponential search in \crefrange{alg:boost:line:doubling}{alg:boost:doubling-end}.
Since \alg{AlgMain} only satisfies its running time guarantee when we are in the \Close{} case, each call in \cref{alg:boost:recurse} needs to be stopped and interrupted if it exceeds this time budget.
Note that each time that \alg{AlgMain} calls \alg{AlgBoosted}, the block periodicity is reduced so the recursion eventually stops.

\begin{algorithm}[t]
\caption{} \label{alg:main}
\begin{algorithmic}[1]
\Procedure{AlgBoosted}{$X,Y, k, K, p, \delta$}
    \For{$\tilde k = 0,1,2,4,\dots,k$} \label{alg:boost:line:doubling}
        \RepeatTimes{$\Theta(\log(\log(k)/\delta))$} \label{alg:boost:line:boost}
            \State Run $\alg{AlgMain}(X, Y, \tilde k, K, p)$ and store the outcome. If it does not finish within \label{alg:boost:recurse}
            \Statex[4] time budget $O\left(\left(\frac{n}{K}\Delta + \tilde{k}K \cdot \Delta^3 \right) \cdot (\log n)^{\alpha \cdot \log_\Delta(n)} \cdot (\log n)^{14 \log_\Delta(p)}\right)$
            \Statex[4] interrupt and store \Far{}
        \EndRepeatTimes \label{alg:boost:doubling-end}
        \If{the majority of the outcomes is \Close{}}
            \Return \Close \label{alg:boost:line:returnclose}
        \EndIf
    \EndFor
    \State \Return \Far{}
\EndProcedure

\medskip

\Procedure{AlgMain}{$X, Y, k, K, p$}
\If{$K > |X| + |Y|$}\label{alg:main:line:testlargek}
    \State \Return \Close
\EndIf
\If{$k = 0$}\label{alg:main:line:testzerok}
    \State \Return output of the Equality Test (\cref{lem:equality-test}) for $X, Y$ with $r = 1/K, \delta := 0.001$ \label{alg:main:line:base-equality}
\EndIf
\If{$p \leq k \Delta$} \label{alg:main:line:base-case-test}
    \State \Return $\alg{AlgSmallBP}(X, Y, k, K, k\Delta, K\Delta)$ (\cref{lem:main-stoc}) \label{alg:main:line:base-case-call}
\Else
    \State Let $X_1,Y_1,\dots,X_s,Y_s$ be the output of $\alg{Split}(X, Y, k, K, 0.01)$\label{alg:main:line:split}
    \For{$d = 1,2,4,\dots,K$} \label{alg:main:line:pslstart}
        \State Sample a set $S \subseteq \intervalcc{1}{s}$ including each element
        \Statex[4] independently with probability $\frac{108 d \log K \log(1/\delta)}{K}$, where $\delta := 0.01/\log K$ \label{alg:main:line:sample}
        \State Run $\alg{AlgBoosted}(X_i, Y_i, d \frac{k}{K} 64 \log K, d, 16 k \log K, 0.01/n^2)$ for all $i \in S$\label{alg:main:line:recursion}
        \If{at least $12\log(1/\delta)$ of the answers are \Far{}} \label{alg:main:line:pslclosecondition}
        \State \Return \Far{} \label{alg:main:line:breaks-end}
    \EndIf
    \EndFor
    \State \Return \Close{} \label{alg:main:line:returnfar}
\EndIf
\EndProcedure

\end{algorithmic}
\end{algorithm}

\paragraph*{Running Time and Correctness Analysis}
Formally, we prove the following lemma:
\begin{lemma}\label{lem:main-alg}
    There is an algorithm that given strings $X, Y$ of total length $n$,
    as well as parameters $k, K, p$ and $\Delta\in \Int_+$ such that: \begin{enumerate*}[label=(\roman*)] \item $p \leq n$, \item $\BP_p(X) \leq \tfrac{K}{k}p$,
        \item $(256 \log K)^2 \leq \Delta \leq n$,  and \item  $K/k > (\log n)^{c \log_\Delta(n)}$ for a sufficiently large constant $c > 0$,
    \end{enumerate*}
    solves the $\GapED(k, K)$ problem with probability at least $2/3$.
    The algorithm runs in time
    \[
        \left(\frac{n}{K}\Delta +\min(E, k)\cdot K\Delta^3 \right) \cdot (\log n)^{O(\log_\Delta(n))},
    \]
    where  $E := \ED(X, Y)$.
\end{lemma}

\begin{proof}
    We prove the lemma using mutual induction over both \alg{AlgBoosted} and \alg{AlgMain}.
    Formally, we have the following inductive hypothesis.

    \paragraph*{Inductive Hypothesis}
    Let $X, Y$ be strings of total length $n$, and let $k, K, p,\Delta\in \Int_+$ be parameters
    such that: \begin{enumerate*}[label=(\roman*)]
    \item $p \leq n$, \item $\BP_p(X) \leq \frac{K}{k}p$, \item $(256 \log K)^2 \leq \Delta$, and
    \item and $K/k \geq (\log n)^{\beta \log_\Delta(n) + 14 \log_\Delta(p)}$, where $\beta > 0$ is the same constant as in \cref{lem:main-stoc}.\end{enumerate*}
    Then, the following holds:
    \begin{enumerate}[label=(\roman*)]
        \item\label{it:boosted} Let $E := \ED(X, Y)$. $\alg{AlgBoosted}(X, Y, k, K, p, \delta)$ solves the $\GapED(k, K)$ problem with probability at least $1 - \delta$ and runs in time
        \[
            O\left(\left(\frac{n}{K}\Delta +\min(E, k) K \cdot \Delta^3 \right) \cdot (\log n)^{\alpha \cdot \log_\Delta(n)} \cdot (\log n)^{14 \log_\Delta(p)} \cdot \log (\log(k)/\delta) \cdot \log k\right).
        \]
        \item\label{it:main} With probability at least 0.9, $\alg{AlgMain}(X, Y, k, K, p)$ solves the $\GapED(k, K)$ problem.
        Moreover, if $\ED(X, Y) \leq k$, then, with probability at least 0.9, it runs in time
        \[
            O\left(\left(\frac{n}{K}\Delta +kK \cdot \Delta^3 \right) \cdot (\log n)^{\alpha \cdot \log_\Delta(n)} \cdot (\log n)^{14 \log_\Delta(p)}\right).
        \]
    \end{enumerate}
    For both running times above, $\alpha > 0$ is the constant in the running time of \cref{lem:main-stoc}.

    \paragraph*{Base Case}
    Let $X, Y$ be strings of total length $n$ and let $k, K, p$ be integers satisfying the conditions of the inductive hypothesis.
    For the base case, we prove statement \ref{it:main} when $n < 4$, $K > n$, $k = 0$ or $p \leq k\Delta$.
    If $n < 4$, then we can solve the problem directly trivially in time $O(1)$ (we omit this trivial case from the pseudocode for brevity).
    If $K > n$, then we can return \Close{} since $\ED(X, Y) \leq n$.
    If $k = 0$, then the task is to distinguish whether $X = Y$ or $\ED(X, Y) > K$.
    Since $\HD(X,Y)\ge \ED(X,Y)$, the Equality Test (see~\cref{lem:equality-test}) correctly solves the problem in \cref{alg:main:line:base-equality} in time $O(n/K)$; this is correct with probability $0.99$.
    Otherwise, since $p \leq k\Delta$, the subproblem is solved directly in \cref{alg:main:line:base-case-call}.
    By assumption, we know that $\BP_{p}(X) \leq K/k \cdot p \leq K\Delta$ and that
    $K/k \geq (\log n)^{\beta \log_\Delta(n) + 14 \log_\Delta(p)}$.
    This means that the call to \alg{AlgSmallBP} is valid.
    Hence, by \cref{lem:main-stoc}, the base case is solved with probability 0.9 in overall time
    \[
        O\left((n/K \Delta + kK \Delta^3) (\log n)^{\alpha \cdot \log_\Delta(n)}\right),
    \]
    which satisfies \ref{it:main}.

    \paragraph*{Inductive Step for \alg{AlgBoosted}}
    Let $X, Y$ be strings of total length $n$ and let $k, K, p$ be integers satisfying the conditions of the inductive hypothesis.
    Suppose that \ref{it:main} holds for any parameters $n' \leq n, k' \leq k, K' \leq K, p' \leq p$ (which satisfy the conditions of the inductive hypothesis).
    We will prove that \ref{it:boosted} holds for $X, Y, k, K, p$.

    Consider the execution of $\alg{AlgBoosted}(X, Y, k, K, p, \delta)$.
    The calls it makes to $\alg{AlgMain}$ in \cref{alg:boost:recurse} satisfy \ref{it:main} as we just argued.
    We start by showing correctness.
    First, suppose that $E := \ED(X, Y) \leq k$.
    Consider any iteration of the loop in \cref{alg:boost:line:doubling} with parameter $\tilde k \leq E$.
    Since $E \leq k$, returning \Close{} is correct, so if the majority of outcomes is \Close{} then we correctly return \Close{}.
    If none of these iterations return \Close{}, consider the iteration when $\tilde{k}/2 < E \leq \tilde{k}$.
    At this iteration, by \ref{it:main}, each call to \alg{AlgMain} correctly returns \Close{} and runs in the time budget with probability at least 0.8.
    Since we take the majority outcome out of $\Theta(\log(\log(k)/\delta))$ repetitions, by Chernoff's bound, we conclude that for this iteration we return \Close{} in \cref{alg:boost:line:returnclose} with probability at least $1 - \delta/\log k \geq 1 - \delta$.

    Now, consider the case when $E > K$.
    Fix some iteration of the loop in \cref{alg:boost:line:doubling}.
    By \ref{it:main}, for each call to \alg{AlgMain} in \cref{alg:boost:recurse}, we correctly store \Far{} with probability at least 0.9 (regardless of whether we interrupt the algorithm or not).
    Thus, by Chernoff's bound, we do not return \Close{} in \cref{alg:boost:line:returnclose} with probability at least $1 - \delta/\log k$.
    By a union bound, we conclude that we do not return \Close{} in any of $O(\log k)$ iterations of the loop in \cref{alg:boost:line:doubling} with probability at least $1 - \delta$.

    Now, we argue about the running time.
    As shown above, if $E \leq k$, then with probability at least $1 - \delta$ we return \Close{} in the iteration when $\tilde{k}/2 < E \leq \tilde{k}$.
    If $E > K$, then we execute all the $\log k$ iterations of the loop in \cref{alg:boost:line:doubling}.
    Thus, the overall running time is bounded by
    \[
        O\left(\left(\frac{n}{K}\Delta + \min(E,k)K \cdot \Delta^3\right) \cdot (\log n)^{\alpha \cdot \log_\Delta(n)} \cdot (\log n)^{14 \log_\Delta(p)} \log(\log(k)/\delta) \cdot \log k \right).
    \]

    \paragraph*{Inductive Step for \alg{AlgMain}}
    Consider strings $X, Y$ of total length $n$ and let $k, K, p$ be integers satisfying the conditions of the inductive hypothesis.
    Suppose that \ref{it:boosted} holds for any parameters $n' \leq n, k' \leq k, K' \leq K$ and $p' < p$ (note that here $p'$ is strictly less than $p$).
    We proceed to prove statement \ref{it:main} for $X, Y, k, K, p$.

    The cases of $n < 4$, $K > n$, $k = 0$ or $p \leq k\Delta$ were handled by the base case, so assume that $n \geq 4$, $K \leq n$, $k \geq 1$ and $p > k\Delta$.
    Therefore, the algorithm continues in \crefrange{alg:main:line:split}{alg:main:line:returnfar} and makes recursive calls to \alg{AlgBoosted}.

    First, we analyze the running time.
    For this, we can assume that $E = \ED(X, Y) \leq k$.
    The call to \alg{Split} in \cref{alg:main:line:split} runs in expected time $O(\frac{n}{K}\log n)$.
    By \cref{lem:split}, with probability at least $1 - \frac{3k}{K}\log(n/0.01)$, its output $X_1,Y_1,\dots,X_s,Y_s$ satisfies $\sum_i \ED(X_i, Y_i) = E$, and $\BP_{k}(X_i) \leq 4K$ for every $i \in \intervalcc{1}{s}$ with probability 0.99.
    By assumption, we have that $K/k \geq (\log n)^{\beta \log_\Delta(n) + 14 \log_\Delta(p)}$ where $\beta > 0$ is a sufficiently large constant.
    Thus, by a union bound, can bound the overall success probability by $1 - 3(k/K) \log(n/0.01) - 0.01 \geq 0.98$.
    From now on, we condition on this event.
    Let $E_i := \ED(X_i, Y_i)$ and $n_i := |X_i| + |Y_i|$.

    \begin{claim}\label{claim:1}
        Every call to $\alg{AlgBoosted}(X_i, Y_i, k', K', p', \delta')$ in \cref{alg:main:line:recursion} satisfies:
        \begin{enumerate}[label=(\arabic*)]
            \item\label{it:pprim} $p' \leq p/\Delta^{1/2}$ and $p' \leq n'$ where $n' = |X_i| + |Y_i|$,
            \item\label{it:bp} $\BP_{p'}(X_i) \leq (K'/k') \cdot p'$,
            \item\label{it:kk} $K'/k' \geq (\log n_i)^{\beta \log_\Delta(n_i) + 14\log_\Delta(p')}$,
            \item\label{it:delta} $\Delta > (256 \log K')^2$.
        \end{enumerate}
    \end{claim}
    \begin{claimproof}
        Observe that, in \cref{alg:main:line:recursion}, the algorithm sets $p' = 256 k \log K$ and
        \[\frac{K'}{k'} = \frac{d}{d (k/K) 64 \log K} = \frac{K}{64 k \log K}. \]
        To show \ref{it:bp}, note that since for each $X_i$ it holds that $\BP_{k}(X_i) \leq 4K$, it follows that $\BP_{p'}(X_i) \leq 4K = (K'/k') \cdot p'$.
        Since $p > k\Delta$ and $\Delta > (256 \log K)^2$, we obtain that $p' < p/\Delta^{1/2}$.
        For the corner case when $256k\log K > n'$, note that setting $p' = \min(n', 256k \log K)$ we still have that $p' < p/\Delta^{1/2}$, and \ref{it:bp} remains valid too (we omitted this case from the pseudocode for readability). Therefore, we obtain \ref{it:pprim}.

        To obtain \ref{it:kk}, we use the assumption that $K/k \geq (\log n)^{\beta \log_\Delta(n) + 14 \log_\Delta(p)}$ and that $K \leq n$ (since the case $K > n$ was handled by the base case), yielding
        \[
            K'/k' = K/(64 k \log K) \geq (\log n)^{\beta \log_\Delta(n) + 14 \log_\Delta(p) - 1}/64.
        \]
        We continue bounding this expression.
        By \ref{it:pprim}, we have that $\log_\Delta(p) \geq \log_\Delta(p' \Delta^{1/2}) = \log_\Delta(p') + 1/2$.
        Hence, using that $64 \leq (\log n)^6$ (the case $n < 4$ was handled by the base case) and $n_i = |X_i| + |Y_i| \leq n$, we obtain that $K'/k' \geq (\log n_i)^{\beta \log_\Delta(n_i) + 14\log_\Delta(p')}$.

        Finally, \ref{it:delta} follows since $K' \leq K$ and by assumption $\Delta > (256 \log K)^2$.
    \end{claimproof}

    To bound the expected running time of \crefrange{alg:main:line:pslstart}{alg:main:line:returnfar}, observe that, for a fixed value of $d$, the probability that a subproblem $X_i, Y_i$ is called in \cref{alg:main:line:recursion} is $\frac{108 d}{K}\log K \log(1/\delta)$, where $\delta = 0.01/\log K$.
    Let $E_i := \ED(X_i, Y_i)$ and $n_i := |X_i| + |Y_i|$.
    \cref{claim:1} implies that we can use the inductive hypothesis \ref{it:boosted} to bound the running time for each call to \alg{AlgBoosted}.
    Hence, the expected running time of one iteration of the loop in \cref{alg:main:line:pslstart} can be bounded as
    \[
        O\left(\sum_{i=1}^s \frac{d}{K} \left(\frac{n_i}{d} \Delta + E_i d \Delta^3\right) (\log n_i)^{\alpha \log_\Delta(n_i)} (\log n_i)^{14 \log_\Delta(p')}\log^2 n \log K \log(1/\delta)\right).
    \]
    Here, we bounded the factors $\log(\log(k')/\delta')\log(k') \leq O(\log^2 n)$ in the running time of \alg{AlgBoosted} given by \ref{it:boosted}, since the parameter $\delta'$ in the call to \alg{AlgBoosted} is set to $\delta' = 0.01/n^2$, and $k' \leq K \leq n$.

    By \cref{claim:1} we have $p' \leq p/\Delta^{1/2}$, and therefore $14\log_\Delta(p') \leq 14\log_\Delta(p) - 5$.
    Since $n_i \leq n$, $K \leq n$, and $1/\delta \leq n^2$, we can bound
    \[
        (\log n_i)^{14 \log_\Delta(p')} \log^2 n \log K \log(1/\delta) \leq O((\log n)^{14 \log_\Delta(p) - 1}).
    \]
    By the guarantees of split, we know that $\sum_i n_i = n$ and $\sum_i E_i = E$.
    Thus, combining the above, we can bound the expected running time of one iteration of the loop in \cref{alg:main:line:pslstart} as
    \[
        O\left(\left(\frac{n}{K}\Delta + EK\Delta^3\right) (\log n)^{\alpha \log_{\Delta}(n)} (\log n)^{14 \log_{\Delta}(p) - 1}\right),
    \]
    where we bounded $Ed^2/K \leq EK$ using that $d \leq K$.
    The overall expected running time across the $O(\log K)$ iterations of the for loop in \cref{alg:main:line:pslstart} adds another $O(\log K) \leq O(\log n)$ factor.
    Hence, we obtain expected time
    \[
        O\left(\left(\frac{n}{K}\Delta + EK\Delta^3\right)(\log n)^{\alpha \log_{\Delta}(n)} (\log n)^{14\log_\Delta(p)}\right).
    \]
    Finally, by Markov's inequality, the algorithm does not exceed 20 times this time bound with probability at least 0.95.
    Together with the initial conditioning on the success of \alg{Split}, we obtain total success probability at least 0.9.
    This completes the proof of the running time for \ref{it:main}.

    \medskip

    Now we argue about the correctness of \crefrange{alg:main:line:pslstart}{alg:main:line:returnfar}.
    Due to \cref{claim:1}, we can use the inductive hypothesis on the calls to \alg{AlgBoosted} in \cref{alg:main:line:recursion}.
    Thus, each call is correct with probability at least $1 - 0.01/n^2 \ge 1-0.01/(s \log K)$ and hence, all the calls are correct with probability at least 0.99.
    The correctness of the algorithm follows from the following claim:
    \begin{claim}\label{claim:precision-sampling}
        Consider an execution of \crefrange{alg:main:line:pslstart}{alg:main:line:breaks-end}.
        Suppose all recursive calls are correct.
        Let $E_i := \ED(X_i, Y_i)$ for $i \in \intervalcc{1}{s}$.
        Let $\calE$ be the event that, if $\sum_i E_i > K$, then we return \Far{}
        and, if $\sum_i E_i \leq k$, then we return \Close{}.
        Then, $\Pr(\calE) \geq 0.99$.
    \end{claim}

    Before proving the claim, let us see how to derive correctness from it.
    As argued in the running time analysis, if $\ED(X, Y) \leq k$ then with probability at least 0.99 the call to \alg{Split} succeeds, and we obtain $\sum_i E_i = \ED(X, Y) \leq k$ (see \cref{lem:split}).
    Then, by \cref{claim:precision-sampling} we return \Close{}.
    To bound the error probability, we take a union bound over all the $O(s \log K)$ calls to $\alg{AlgBoosted}$, the call to $\alg{Split}$, and \cref{claim:precision-sampling}.
    All calls to \alg{AlgBoosted} succeed with probability at least $0.99$, the call to \alg{Split} succeeds with probability 0.99, and \cref{claim:precision-sampling} succeeds with probability 0.99.
    Thus, the overall success probability is at least 0.95.
    On the other hand, if $\ED(X, Y) > K$, then by the subadditivity of edit distance, the partition returned by \alg{Split} satisfies $\sum_i E_i \geq \ED(X, Y) > K$ (deterministically).
    Thus, assuming that the calls to \alg{AlgBoosted} are correct, \cref{claim:precision-sampling} guarantees that we return \Far{}.
    Similarly as before, the success probability is at least 0.95.
    It remains to prove the claim.

    \begin{claimproof}[Proof of \cref{claim:precision-sampling}]
        First, consider the case where $\sum_i E_i > K$.
        For $\ell \ge 0$, define \[L_\ell := \set{j \mid 2^\ell \leq \min(E_j, K+1) < 2^{\ell + 1}}.\]
        Thus, note that there are at most $T := \floor{\log(K+1)} + 1$ non-empty sets $L_\ell$.
        The key observation is that since $\sum E_i > K$, there exists a level $\ell^*$ such that $|L_{\ell^*}| \geq K/(2^{\ell^*+1} T)$.
        Indeed, note that otherwise we obtain
        \[
          K \leq \sum_{j=1}^s \min(E_j, K+1)
            \leq \sum_{\ell=0}^{\floor{\log(K+1)}} 2^{\ell+1}|L_\ell|
            < \sum_{\ell = 0}^{\floor{\log(K+1)}} 2^{\ell+1} \frac{K}{2^{\ell + 1} T} = K;
        \]
        a contradiction.

        Fix such a level $\ell^*$ and focus on the iteration of the loop in \cref{alg:main:line:pslstart} when $d = 2^{\ell^*}$.
        For each $i \in L_{\ell^*}$, let $I_i$ be the indicator random variable which equals 1 if $i$ is included in the sample $S$.
        Let $I := \sum_{i \in L_{\ell^*}}I_i$.
        Since we sample $S$ with rate $\frac{108 d}{K}\log K \log(1/\delta)$, it follows that
        \[
            \Ex(I) = |L_{\ell^*}| \cdot \frac{108 d}{K} \log K \log(1/\delta) \geq \frac{K}{2 d T} \cdot \frac{108 d}{K}\log K \log(1/\delta)
            = \frac{54}{T}\log K \log (1/\delta).
        \]
        Since $T = \floor{\log(K+1)}+1 \leq \log(K) + 2 \leq 3\log K$, we obtain that $\Ex(I) \geq 18 \log(1/\delta)$.
        Note that the algorithm correctly returns \Far{} if $I \geq 12\log(1/\delta)$ (see \cref{alg:main:line:pslclosecondition}).
        By a Chernoff bound, we bound the error probability by
        \[
            \Pr(I < 12\log(1/\delta)) = \Pr(I < (1 - 1/3)\Ex(I)) \leq \exp(-18\log(1/\delta)/(2 \cdot 9)) = \delta.
        \]

        Now, consider the case when $\sum_i E_i \leq k$.
        Let $g := K/(64 \cdot k \log K)$ be the gap which the recursive call to \alg{AlgBoosted} distinguishes.
        Define $\tilde L_\ell := \set{j \mid E_j \geq 2^\ell/g}$.
        Note that for all $\ell$, it holds that $|\tilde{L}_\ell| \leq K/(64 \cdot 2^\ell \log K)$.
        Indeed, if there was some $\ell$ with $|\tilde{L}_\ell| > K/(64 \cdot 2^\ell \log K)$ then
        \[
            k \geq \sum_{i=1}^s E_i \geq |\tilde{L}_\ell| \frac{2^\ell}{g} > \frac{K}{64 g \log K} = k;
        \]
        a contradiction.
        Focus on iteration $d = 2^{\ell}$ of the for-loop in \cref{alg:main:line:pslstart}.
        For every $i \in \tilde{L}_\ell$, define the indicator random variable $I_i$ which equals one if $i \in S$.
        Let $I := \sum_{i \in \tilde{L}_\ell} I_i$.
        Since we sample $S$ with rate $\frac{108 d}{K}\log K \log(1/\delta)$, it follows that
        \[
            \Ex(I)
                = |\tilde{L}_{\ell}| \cdot \frac{108d}{K} \log K \log(1/\delta)
                \leq \frac{K}{64 d \log K} \cdot \frac{108 d}{K}\log K \log(1/\delta) \leq 1.7 \log(1/\delta).
        \]
        Note that the set $\tilde L_\ell$ contains all subproblems $E_i$ for which \alg{AlgBoosted} may return \Far{}.
        Thus, the algorithm correctly returns \Close{} if for every iteration $I < 12 \log(1/\delta)$ (see \cref{alg:main:line:pslclosecondition}).
        By a Chernoff bound we can bound this error probability as
        \[
            \Pr(I \ge 12 \log(1/\delta)) \leq \Pr(I \ge 7\Ex(I)) \leq 2^{-7 \Ex(I)} \leq \exp(-\log(1/\delta)) = \delta.
        \]
        The claim follows by taking a union bound over the $\log K$ iterations and recalling our choice of $\delta = 0.01/\log K$.
    \end{claimproof}

    Finally, note that the lemma statement follows by the inductive hypothesis \ref{it:boosted}, that is, by calling \alg{AlgBoosted} with error probability $\delta := 1/3$.
\end{proof}

\subsection{Main Theorem}\label{sec:main:sec:main-thms}
We now put things together to prove \cref{thm:main} and its corollaries.

\thmmain*
\begin{proof}
    We run \cref{lem:main-alg} with parameters $X, Y, k, K, \Delta' := \max(\Delta, (256 \log K)^2)$ and $p := n$.
    Note that $\BP_p(X) = 1 \leq p \cdot K/k$ holds, so the call is valid.
    If $\Delta \geq (256 \log K)^2$, the lemma follows immediately by the guarantees of \cref{lem:main-alg}.
    If $2 \leq \Delta < (256 \log K)^2$, then observe that by setting $\Delta' = (256 \log K)^2$ the gap and the running time of \cref{lem:main-alg} do not increase compared to the lemma statement, which completes the proof.
\end{proof}

\corsubpoly*
\begin{proof}
    If $k^6 < n$, run the algorithm of \cite[Corollary 2]{BringmannCFN22} to solve the problem directly in time $O(n/k)$.
    So we can assume that $n \leq k^6$.
    In particular, $\log n = \Theta(\log k)$.
    We use \cref{thm:main} setting $\Delta = 2^{\sqrt{\log k \log\log k}}$.
    This gives an algorithm for the $(k, K)$-gap problem
    running in time $O\left((n/K + kK) \cdot 2^{c_0 \cdot \sqrt{\log k \log\log k})}\right)$ for some constant $c_0 > 0$,
    where $K/k \geq 2^{c \sqrt{\log k \log \log k}}$ with $c > 0$ sufficiently large.
    Thus, setting $K = k \cdot 2^{\mu \cdot \sqrt{\log k \log \log k}}$ where $\mu > c_0 + c$ yields the result.
\end{proof}

\corpolylog*
\begin{proof}
    If $k^6 < n$, we solve the problem directly in time $O(n/k^{1-\varepsilon})$ using the algorithm of \cite[Corollary 3]{BringmannCFN22}.
    Thus, we can assume from now on that $n \leq k^6$ and therefore $\log n = \Theta(\log k)$.
    We use \cref{thm:main} choosing $\Delta := k^{\delta/3}$ for some $0 < \delta < \varepsilon$, and $K = k(\log k)^{O(1/\varepsilon)}$ with large enough hidden constant.
    This yields an algorithm for the desired gap running in time
    \[(n/k + k^2) \cdot k^{\delta} \cdot (\log k)^{O(1/\delta)} \leq O(n/k^{1-\varepsilon} + k^{2+\varepsilon}).\qedhere\]
\end{proof}

\corpoly*
\begin{proof}
    If $K < 2^{(\log n)^{2/3}}$, then we can solve the $(k, K)$-gap edit distance problem exactly using the Landau Vishkin algorithm in time $O(n + k^2) \subseteq O(n^{1+o(1)}/K + k^2)$.
    Hence, since $K > k^{1+\varepsilon}$ we can assume from now on that $K/k \geq 2^{\Theta((\log n)^{2/3})}$.

    Instantiating \cref{thm:main} with $\Delta = 2^{\sqrt{\log n \log\log n}}$, we obtain an algorithm that can distinguish the gap $\tilde{K}/k \geq 2^{\Theta(\sqrt{\log n \log \log n})}$ in time $O((n/\tilde{K} + k\tilde{K}) \cdot n^{o(1)})$.
    If \smash{$K < \sqrt{n/k}$}, then running this algorithm with \smash{$\tilde{K} := K$} solves the $(k, K)$-gap problem in time $O(n^{1+o(1)}/K)$.
    (Here we used the assumption that $K/k \geq 2^{\Theta((\log n)^{2/3})}$, since this is larger than the minimum gap distinguishable by \cref{thm:main}.)
    Otherwise, note that running this algorithm with $\smash{\tilde{K} := \max(\sqrt{n/k}, k \cdot 2^{\Theta(\sqrt{\log n \log \log n})})}$ correctly solves the $(k, K)$-gap problem, since we do not increase the gap.
    In this case, the running time is $O((\sqrt{nk} + k^2) \cdot n^{o(1)})$.
    Combining the above yields the claimed running time.
\end{proof}

\section{Algorithm for Bounded Block Periodicity}\label{sec:bcfn}

In this section, we give the proof of \cref{lem:main-stoc}.
The algorithm closely follows the approach of Bringmann, Cassis, Fischer and Nakos~\cite{BringmannCFN22}, which builds upon the \emph{tree distance framework} pioneered by Andoni, Krauthgamer and Onak \cite{AndoniKO10}.

\subsection{Tree Distance Framework}

We start by recalling several definitions and tools for the general setup of the \cite{BringmannCFN22} algorithm.
The main ingredient is a way to split the computation of the edit distance into independent subtasks.
This is achieved via the \emph{tree distance}, which was introduced in \cite{AndoniKO10}.
To define the tree distance, we first need an underlying \emph{partition tree}.

\begin{definition}[Partition Tree] \label{def:partition-tree}
    Let $X$ and $Y$ be length-$n$ strings.
    A partition tree $T$ for $X$ and $Y$ is a balanced $\ell$-ary tree with $n$ leaves numbered from $0$ to $n-1$ (from left to right).
    Each node $v$ in $T$ is associated with a multiplicative accuracy $\alpha_v > 1$ and a rate $r_v \geq 0$.

    For each node $v$ in $T$, we define the substring $X_v$ as follows: If the subtree below $v$ spans from the $i$-th to the $j$-th leaf, then we set $X_v = X\intervalco{i}{j}$.
    Similarly, for a shift $s \in \Int$, we set $Y_{v, s} = Y \intervalco{i + s}{j + s}$ (the substring of $Y$ relevant at $v$ for one specific shift $s$).
\end{definition}

With this definition at hand, we can define the shift-restricted tree distance:

\begin{definition}[Shift-Restricted tree distance]\label{def:shift-td}
    Let $T$ be a partition tree for length-$n$ strings $X$ and $Y$, and let $L\ge 0$ be an integer.
    For every node $v$ in $T$ and every shift $s \in \intervalcc{-L}{L}$, we define the $L$-restricted tree distance
    $\TD_{v,s}^L(X, Y)$ as follows:
    \begin{itemize}
        \item If $v$ is a leaf, then $\TD_{v,s}^L(X, Y) = \ED(X_v, Y_{v,s})$.
        \item If $v$ is a node with children $v_0, \dots, v_{\ell-1}$, then
        \begin{equation}\label{eq:td}
            \TD_{v,s}^L(X, Y) = \sum_{i = 0}^{\ell-1} \TTD_{v_i,s}^L(X,Y),
        \end{equation}
        where
        \begin{equation}\label{eq:ttd}
        \TTD_{v_i,s}^L(X,Y) = \min_{s' \in \intervalcc{-L}{L}} \left(\TD^L_{v_i, s'}(X, Y) + 2|s - s'|\right).
        \end{equation}
    \end{itemize}
\end{definition}

Since we restrict to shifts in $\intervalcc{-L}{L}$, we define the substring of $Y$ relevant at a node $v$ as $Y_v := Y\intervalco{i-L}{j+L}$.

The difference between \cref{def:shift-td} and the tree distance definition in \cite{AndoniKO10,BringmannCFN22} is that we restrict the shifts to the set $\intervalcc{-L}{L}$.
The following lemma captures the relationship between the tree distance and edit distance.

\begin{lemma}\label{lem:equiv-td-ed}
    Let $T$ be a partition tree for length-$n$ strings $X$ and $Y$, and let $L\ge 0$ be an integer.
    Suppose that $T$ has maximum degree $\ell$ and height (the maximum distance from the root to a leaf) at most $h$.
    Then, the $\TD^{L}(X,Y)$ can be bounded as follows:
    \begin{itemize}
        \item  $\TD^L(X, Y) \ge \ED(X,Y)$;
        \item  $\TD^L(X, Y) \le (2 (\ell-1) h+1) \ED(X, Y)$ provided that $\ED(X, Y)\le L$.
    \end{itemize}
\end{lemma}
\begin{proof}
    For the lower bound, we inductively prove that
    \[\ED(X_v,Y_{v,s})\le \TD^L_{v,s}(X,Y) \qquad\text{and}\qquad \ED(X_v,Y_{v,s})\le \TTD^L_{v,s}(X,Y)\]
    hold for every node $v$ and every shift $s\in \intervalcc{-L}{L}$.

    The first claim holds trivially if $v$ is a leaf.
    Otherwise, we have $X_v = \bigodot_{i=0}^{\ell-1} X_{v_i}$ as well as $Y_{v,s} = \bigodot_{i=0}^{\ell-1}Y_{v_i,s}$.
    The subadditivity of edit distance and the inductive assumption imply
    \[\ED(X_v,Y_{v,s})\le \sum_{i=0}^{\ell-1} \ED(X_{v_i},Y_{v_i,s}) \le \sum_{i=0}^{\ell-1} \TTD^L_{v_i,s}(X,Y) = \TD^L_{v,s}(X,Y).\]
    To prove the second claim, observe that the following holds for every $s'\in \intervalcc{-L}{L}$:
    \[\ED(X_v,Y_{v,s}) \le \ED(X_v,Y_{v,s'}) + \ED(Y_{v,s'},Y_{v,s}) \le \TD^L_{v,s'}(X,Y) + 2|s-s'|.\]
    Consequently,
    \[\ED(X_v,Y_{v,s}) \le \min_{s'\in \intervalcc{-L}{L}} \left(\TD^L_{v,s'}(X,Y) + 2|s-s'|\right) = \TTD^L_{v,s}(X,Y).\]

\newcommand{\symdif}{\mathbin{\triangle}}
\newcommand{\sub}{\subseteq}

    For the upper bound, we prove the following claim for every node $v$ and every shift $s\in \intervalcc{-L}{L}$.
    Denote by $h_v$ the height of the subtree rooted at $v$ (that is, the maximum distance from $v$ to a descendant of $v$).
    For two fragments $Y\intervalco{p}{q}$ and $Y\intervalco{p'}{q'}$, denote $|Y\intervalco{p}{q}\symdif Y\intervalco{p'}{q'}| = |p-p'|+|q-q'|$.

    \begin{claim}
        Let $Y'_{v}$ be a fragment of~$Y$ such that $\ED(X_v,Y'_{v}) +|Y_{v,0} \symdif Y'_{v}| \le L$. Then,
         \begin{equation}\label{eq:ttdub}\TTD^L_{v,s}(X,Y) \le (2(\ell-1) h_v+1)  \ED(X_v,Y'_{v}) + |Y_{v,s} \symdif Y'_{v}|.\end{equation}
         Moreover, if $|Y_{v,s} \symdif Y'_{v}|=\big||Y_{v,s}|-|Y'_{v}|\big|$, i.e., one of the fragments is contained in the other, then
         \begin{equation}\label{eq:tdub}\TD^L_{v,s}(X,Y) \le (2(\ell-1) h_v+1)  \ED(X_v,Y'_{v}) + |Y_{v,s} \symdif Y'_{v}|.\end{equation}
    \end{claim}

    Let us first prove \eqref{eq:tdub}. If  $v$ is a leaf, then $h_v=0$ and simply
    \[\TD^L_{v,s}(X,Y) = \ED(X_v,Y_{v,s}) \le \ED(X_v,Y'_{v}) + \ED(Y'_{v},Y_{v,s})\le  \ED(X_v,Y'_{v}) +  |Y_{v,s}\symdif Y'_{v}|.\]

    Next, suppose that $v$ has children $v_i$ for $i\in \intervalco{0}{\ell}$.
    Decompose $X_v = \bigodot_{i=0}^{\ell-1} X\intervalco{x_i}{x_{i+1}}$
    so that $X_{v_i} = X\intervalco{x_i}{x_{i+1}}$
    and $Y_{v,s} = \bigodot_{i=0}^{\ell-1} Y\intervalco{y_i}{y_{i+1}}$
    so that $Y_{v_i,s} = Y\intervalco{y_i}{y_{i+1}}$.
    Moreover, decompose $Y'_{v}=\bigodot_{i=0}^{\ell-1} Y\intervalco{y'_i}{y'_{i+1}}$,
    denoting $Y'_{v_i} = Y\intervalco{y'_i}{y'_{i+1}}$,
    so that $\ED(X_v,Y'_{v})=\sum_{i=0}^{\ell-1} \ED(X_{v_i},Y'_{v_i})$.

    We shall inductively apply \eqref{eq:ttdub} for $X_{v_i,s}$ and $Y'_{v_i}$.
    For this, note that
    \begin{align*} \ED(X_{v_i},Y'_{v_i})+|Y_{v_i,0} \symdif Y'_{v_i}|
        &= \ED(X_{v_i},Y'_{v_i}) + |x_i-y'_i| + |x_{i+1}-y'_{i+1}|\\
        &\le \ED(X_{v_i},Y'_{v_i}) + |(x_i-x_0)-(y'_i-y'_0)| +|x_0-y'_0|\\
        &\qquad\qquad\qquad + |(x_\ell-x_{i+1})-(y'_\ell-y'_{i+1})| + |x_{\ell}-y'_{\ell}|\\
    & \le \ED(X\intervalco{x_i}{x_{i+1}},Y\intervalco{y'_i}{y'_{i+1}}) + \ED(X\intervalco{x_0}{x_{i}},Y\intervalco{y'_0}{y'_{i}})\\
    & \qquad\qquad + |x_0-y'_0| +  \ED(X\intervalco{x_{i+1}}{x_{\ell}},Y\intervalco{y'_{i+1}}{y'_{\ell}}) + |x_{\ell}-y'_{\ell}|\\
    &= \ED(X_v,Y'_{v}) + |Y_{v,0} \symdif Y'_{v}|\\
    &\le L\end{align*}
    Consequently, \eqref{eq:ttdub} yields
    \[\TTD^L_{v_i,s}(X,Y) \le (2(\ell-1) h_{v_i}+1)\ED(X_{v_i},Y'_{v_i})+|y_{i}-y'_i|+|y_{i+1}-y'_{i+1}|.\]
    The assumption $|Y_{v,s} \symdif Y'_{v}|=\big||Y_{v,s}|-|Y'_{v}|\big|$
    translates to $|y_0-y'_0|+|y_\ell-y'_\ell| = |(y_\ell-y_0)-(y'_\ell-y'_0)|$.
    Thus, the following holds for every $i\in \intervalcc{0}{\ell}$:
    \begin{align*}2|y_{i}-y'_i|
        &=|(y_i-y_0)-(y'_i-y'_0)+(y_0-y'_0)|+ |(y_i-y_\ell)-(y'_i-y'_\ell)+(y_\ell-y'_\ell)|\\
        &\le |(y_i-y_0)-(y'_i-y'_0)| + |y_0-y'_0| + |(y_i-y_\ell)-(y'_i-y'_\ell)|+|y_\ell-y'_\ell|\\
        &= |(y_i-y_0)-(y'_i-y'_0)| + |(y_\ell-y_i)-(y'_\ell-y'_i)|+|(y_\ell-y_0)-(y'_\ell-y'_0)|\\
        &= |(x_i-x_0)-(y'_i-y'_0)| + |(x_\ell-x_i)-(y'_\ell-y'_i)|+|(x_\ell-x_0)-(y'_\ell-y'_0)|\\
    &\le \ED(X\intervalco{x_0}{x_i},Y\intervalco{y'_0}{y'_i}) + \ED(X\intervalco{x_{i}}{x_{\ell}},Y\intervalco{y'_{i}}{y'_{\ell}})+ \ED(X\intervalco{x_{0}}{x_{\ell}},Y\intervalco{y'_{0}}{y'_{\ell}})\\
    &= 2\ED(X_v,Y'_{v}).\end{align*}
    Summing up over $i\in \intervaloo{0}{\ell}$, due to $h_{v_i}\le h_v-1$, we obtain
    \begin{align*}\TD^L_{v,s}(X,Y) &= \sum_{i=0}^{\ell-1}\TTD^L_{v_i,s}(X,Y)\\
        & \le \sum_{i=0}^{\ell-1}\left( (2(\ell-1) h_{v_i}+1)\ED(X_{v_i},Y'_{v_i})+|y_{i}-y'_i|+|y_{i+1}-y'_{i+1}|\right)\\
        &\le (2(\ell-1)(h_v-1) + 1)\ED(X_v,Y'_{v})+2(\ell-1) \ED(X_v,Y'_{v})+ |y_0-y'_0|+|y_\ell-y'_\ell| \\
        &= (2(\ell-1) h_v + 1)\ED(X_v,Y'_{v}) + |Y_{v,s}\symdif Y'_{v}|.\end{align*}

    It remains to argue that \eqref{eq:ttdub} follows from \eqref{eq:tdub}.
    Due to $\TTD_{v,s}^L(X,Y) \le \TD_{v,s}^L(X,Y)$, this is immediate if $|(y_\ell-y_0)-(y'_\ell-y'_0)|=\big||Y_{v,s}|-|Y'_{v}|\big|=|Y_{v,s} \symdif Y'_{v}|=|y_0-y'_0|+|y_\ell-y'_\ell|$.
    Thus, we henceforth assume $\big||Y_{v,s}|-|Y'_{v}|\big|=\big||y_0-y'_0|-|y_\ell-y'_\ell|\big|$.
    By symmetry, we may also assume without loss of generality that $|y_0-y'_0| \le |y_\ell-y'_\ell|$,
    which implies $\big||Y_{v,s}|-|Y'_{v}|\big| = |y_{\ell}-y'_\ell|-|y_0-y'_0|$.

    In this case, we set $s' := y'_0-x_0$.
    Note that $|s'| = |x_0-y'_0| \le |Y_{v,0} \symdif Y'_{v}| \le L$,
    and thus $s'\in \intervalcc{-L}{L}$.
    Moreover,
    \begin{align*}
        |Y_{v,s'} \symdif Y'_{v}| &= |(x_0+s')-y'_0|+|(x_\ell+s')-y'_\ell|\\
         &= |(x_\ell+s'-x_0-s')-(y'_\ell-y'_0)| \\
         &= \big||Y_{v,s'}|-|Y'_{v}|\big| \\
         &= \big||Y_{v,s}|-|Y'_{v}|\big| \\
         &= |y_{\ell}-y'_\ell|-|y_0-y'_0| \\
         &= |y_0-y'_0| + |y_{\ell}-y'_\ell| - 2|y_0-y'_0| \\
         &= |Y_{v,s} \symdif Y'_{v}|-2|s-s'|.
    \end{align*}
    In particular, $|Y_{v,s'} \symdif Y'_{v}|=\big||Y_{v,s'}|-|Y'_{v}|\big|$ lets us apply \eqref{eq:tdub} for $s'$,
    and thus
    \begin{align*}\TTD^L_{v,s}(X,Y)
         &\le \TD^L_{v,s'}(X,Y) + 2|s-s'| \\
         &\le (2(\ell-1)h_v+1)\ED(X_v,Y'_{v})+|Y_{v,s'}\symdif Y'_{v}| + 2|s-s'|\\
         &=  (2(\ell-1)h_v+1)\ED(X_v,Y'_{v})+|Y_{v,s}\symdif Y'_{v}|\end{align*}
    holds as claimed.
\end{proof}

For our purpose, we will be using the tree distance with shifts restricted in $\intervalcc{-k}{k}$, i.e. setting $L := k$ in \cref{def:shift-td}.

We are now ready to formally define the computational problem we are aiming to solve.
Intuitively, we aim to compute the tree distance as defined in \cref{def:shift-td}.
However, since we are aiming for sublinear time, we can only afford to approximate it.
Recall that each node $v$ in a partition tree $T$ has an associated multiplicative accuracy $\alpha_v > 0$ and a rate $r_v \geq 0$.
The task for each node in $v$ is the following.

\begin{definition}[Tree Distance Problem]\label{def:td-problem}
    Let $\mu > 0$ be a parameter to be set.
    For every node $v$ in the partition tree $T$, compute numbers $\eta_{v,-k},\dots,\eta_{v,k}$ such that
    \begin{equation}
        \frac{1}{\alpha_v}\ED(X_v, Y_{v,s}) - \frac{1}{r_v} \leq \eta_{v,s} \leq \alpha_v \TD_{v,s}^k(X, Y) + \frac{\mu}{r_v}. \label{eqn:td-problem}
    \end{equation}
\end{definition}

\subsection{BCFN Tools}
To present the algorithm that solves the Tree Distance Problem (\cref{def:td-problem}), we start by introducing the necessary tools.

Suppose we are at some node $v$ in the partition tree, and we have already computed the values $\eta_{w,\cdot}$ for all its children $w$ as per \cref{def:td-problem}.
We want to use these approximations to compute the values $\eta_{v,\cdot}$ at $v$ using the recursive definition of the tree distance given by \cref{def:shift-td}.
If we do this naively, the additive error across the children adds up, which is prohibitively high.
In order to control it, we use the following tool known as the \emph{precision sampling lemma} as introduced by Andoni, Krauthgamer and Onak~\cite{AndoniKO10}.

\begin{lemma}[Precision Sampling Lemma~{{{\cite{Andoni17}}}}]\label{lem:precision-sampling}
    Fix parameters $\delta, \varepsilon > 0$. Let $\alpha \geq 1$ and $\beta \geq 0$.
    There is a distribution $\calD = \calD(\varepsilon, \delta)$ supported over $(0,1]$ from which samples can be drawn in expected time $O(1)$, that satisfies the following:
    \begin{description}
        \item[Accuracy] Let~\makebox{$a_1, \dots, a_\ell \geq 0$} be reals, and independently sample~$u_1, \dots, u_\ell \sim \calD$.
        There is an $\Order(\ell \cdot \varepsilon^{-2} \log(\delta^{-1}))$-time algorithm \alg{Recover} satisfying for all $\widetilde a_1, \dots, \widetilde a_\ell$, with success probability at least $1 - \delta$:
        \begin{itemize}
            \item If $\widetilde a_i \geq \frac1\alpha \cdot a_i - \beta \cdot u_i$ for all~$i$,
            then $\alg{Recover}(\widetilde a_1, \dots, \widetilde a_\ell, u_1, \dots, u_\ell) \geq \frac{1}{(1+\varepsilon)\alpha} \cdot \sum_i a_i - \beta$.
            \item If $\widetilde a_i \leq \alpha \cdot a_i + \beta \cdot u_i$ for all~$i$,
            then $\alg{Recover}(\widetilde a_1, \dots, \widetilde a_\ell, u_1, \dots, u_\ell) \leq (1+\varepsilon) \alpha \sum_i a_i + \beta$.
        \end{itemize}
        \item[Efficiency] Sample $u \sim {\cal D}$.
        Then, for any $N \geq 1$ there is an event $\calE = \calE(u)$ happening with probability at least $1 - 1/N$, such that
        $\Ex_{u \sim \cal D}(\,1/u \mid \calE\,) \leq \Order(\varepsilon^{-2} \polylog(N, \delta^{-1}, \varepsilon^{-1}))$.
    \end{description}
\end{lemma}

To efficiently evaluate the tree distance at some node $v$ from the values of its children, we need the following lemma.

\begin{lemma}[{\cite[Lemma 12]{BringmannCFN22}}]\label{lem:range-minimum}
    There is an $\Order(k)$-time algorithm for the following problem: Given integers $A_{-k}, \dots, A_k$, compute for all $s \in \intervalcc{-k}{k}$:
    \begin{equation*}
        B_s = \min_{s' \in \intervalcc{-k}{k}} A_{s'} + 2 \cdot |s - s'|.
    \end{equation*}
\end{lemma}

In order to obtain sublinear time, the key idea used in \cite{BringmannCFN22} is to prune the computation of the tree distance once we reach nodes where we can easily infer the desired $\eta$ values.
To this end, they designed the following property testers.

We say that a node $v$ in the partition tree is \emph{matched} if there is a shift $s^* \in \intervalcc{-k}{k}$ such that $X_v = Y_{v,s^*}$.

\begin{lemma}[{Matching Test \cite[Lemma 18]{BringmannCFN22}}] \label{lem:alignment-test}
    Let $X, Y$ be strings such that $|Y| = |X| + 2k$, and let $r > 0$ be a sampling rate.
    There is an algorithm which returns one of the following two outputs:
    \begin{itemize}
        \item \Close[s^*], where $s^* \in \intervalcc{-k}{k}$ satisfies $\HD(X, Y \intervalco{k + s^*}{|X| + k + s^*}) \leq 1/r$.
        \item \Far, in which case there is no $s^* \in \intervalcc{-k}{k}$ with $X = Y \intervalco{k + s^*}{|X| + k + s^*}$.
    \end{itemize}
    The algorithm runs time $\Order(r |X| \log(\delta^{-1}) + k \log |X|)$ and is correct with probability $1 - \delta$.
\end{lemma}

\begin{lemma}[{$p$-Periodicity Test \cite[Lemma 17]{BringmannCFN22}}] \label{lem:periodicity-test}
    Let $X$ be a string, let $p \geq 1$ be an integer parameter and let $r > 0$ be a sampling rate.
    There is an algorithm which returns one of the following two outputs:
    \begin{itemize}
        \item \Close[Q], where $Q$ is a primitive string of length $\leq p$ with $\HD(X, Q^* \intervalco{0}{|X|}) \leq 1/r$.
        \item \Far, in which case $X$ is not $p$-periodic.
    \end{itemize}
    The algorithm runs in time~$\Order(r |X| \log(\delta^{-1}) + p)$ and is correct with probability $1 - \delta$.
\end{lemma}

\subsection{Algorithm}

Now we are ready to present the algorithm.
The idea is to approximately evaluate the tree distance (\cref{def:shift-td}) over the partition tree.
In order to obtain sublinear running time, we want to avoid recursing over the entire tree.
For that end, we use the insight that whenever we encounter a node $v$ that is matched and $p$-periodic (which we can test using \cref{lem:alignment-test,lem:periodicity-test}), we can stop the recursive computation and approximate the values $\eta_{v,s}$ directly for all shifts $s \in \intervalcc{-k}{k}$, as captured by \cref{lem:fast-shifts-ed}.
The rest of the algorithm combines the recursive results from the children using \cref{lem:precision-sampling} and \cref{lem:range-minimum}.
Consult \cref{alg:stoc} for the full pseudocode.

\begin{algorithm}[ht]
\caption{} \label{alg:stoc}
\begin{algorithmic}[1]
\Input{Strings $X, Y$, a node $v$ in the partition tree $T$, integers $k, K, p \geq 0$ and a rate $r_v$}
\Output{$\eta_{v, s}$ for all shifts $s \in \intervalcc{-k}{k}$}
\medskip

\If{$v$ is a leaf} \label{alg:stoc:line:test-short-condition}
    \State Compute and \Return $\eta_{v,s} = \ED(X_v, Y_{v,s})$ for all $s \in \intervalcc{-k}{k}$ \label{alg:stoc:line:test-short}
\EndIf
\State Run the Matching Test (\cref{lem:alignment-test}) for $X_v, Y_v$ (with $r = 3r_v$ and $\delta = 0.01/n$) \label{alg:stoc:line:alignment-test}
\State Run the $p$-Periodicity Test (\cref{lem:periodicity-test}) for $Y_v$ (with $r = 3r_v$ and $\delta = 0.01/n$) \label{alg:stoc:line:periodicity-test}
\If{the Matching Test returns \Close[s^*]} \label{alg:stoc:line:alignment-test-condition}
    \If{the Periodicity Test returns \Close[Q]} \label{alg:stoc:line:periodicity-test-condition}
        \State Compute and \Return $\eta_{v,s}$ for all $s \in \intervalcc{-k}{k}$ using \cref{lem:fast-shifts-ed} \label{alg:stoc:line:periodic-solve}
    \EndIf
\EndIf

\medskip
\ForEach{$i \in \intervalco{0}{\ell}$} \label{alg:stoc:line:iter-children}
    \State Let $v_i$ be the $i$-th child of $v$ and sample $u_{v_i} \sim \calD(\varepsilon := (2 \log n)^{-1}, \delta := 0.01 /(k n))$ \label{alg:stoc:line:precision}
    \State Recursively compute $\eta_{v_i, s}$ with rate $r_v / u_{v_i}$ for all $s \in \intervalcc{-k}{k}$ \label{alg:stoc:line:recursion}
    \State Compute $\widetilde A_{i, s} = \min_{s' \in \intervalcc{-k}{k}} \eta_{v_i, s'} + 2 \cdot |s - s'|$ using \cref{lem:range-minimum} \label{alg:stoc:line:range-minimum}
\EndForEach
\ForEach{$s \in \intervalcc{-k}{k}$} \label{alg:stoc:line:iter-output}
    \State $\eta_{v,s} = \alg{Recover}(\widetilde A_{0,s}, \dots, \widetilde A_{\ell-1, s}, u_{v_0},\dots,u_{v_{\ell-1}})$ \label{alg:stoc:line:recovery}
\EndForEach
\State\Return $\eta_{v, s}$ for all $s \in \intervalcc{-k}{k}$ \label{alg:stoc:line:return}
\end{algorithmic}
\end{algorithm}

We start by proving \cref{lem:fast-shifts-ed}.
To achieve that, we provide an algorithm that can efficiently approximate the edit distance of strings that are periodic.
More precisely, we leverage the following result whose proof is deferred to \cref{sec:shifts}.

\begin{restatable}{lemma}{lemapproxexperiodic}\label{lem:approxed-periodic-PM}
    Given a positive integer $p$, two strings $P$ and $T$ with period $p$ and lengths $m\le n$, respectively,
    and an integer $\Delta\in \intervalcc{2}{p}$,
    one can compute, for all $i \in \intervalcc{0}{n-m}$, multiplicative $(\log p)^{\Oh(\log_\Delta p)}$-approximations of $\ED(P, T\intervalco{i}{i+m})$
    in time $\Oh(n-m)+p\Delta \cdot (\log p)^{O(\log_\Delta p)}$ correctly with high probability.
\end{restatable}

\begin{restatable}[Fast Shifted ED for Periodic Strings]{lemma}{lemfastshifts}\label{lem:fast-shifts-ed}
    Let $v$ be a node for which the Matching Test correctly returns $\Close[s^*]$ and the Periodicity Test correctly returns $\Close[Q]$ (i.e., \crefrange{alg:stoc:line:alignment-test-condition}{alg:stoc:line:periodicity-test-condition} succeed).
    Then, with high probability, given integers $p\ge |Q|$ and $\Delta\in \intervalcc{2}{p}$, we can compute values $\eta_{v,s}$ for all $s \in \intervalcc{-k}{k}$
    such that
    \[
        \ED(X_v, Y_{v,s})-\tfrac{1}{r_v}\leq \eta_{v,s} \leq (\log p)^{O(\log_\Delta(p))} \cdot \left(\ED(X_v, Y_{v,s})+\tfrac{1}{r_v}\right),
    \]
    in total time $O\left(p(\log p)^{O(\log_\Delta(p))}\Delta + k\right)$.
\end{restatable}
\begin{proof}
    Denote $T=Q^*\intervalco{0}{|Y_v|}$ and $S=T\intervalco{k-s^*}{|X_v|+k-s^*}$.
    Since the Periodicity Test in \cref{alg:stoc:line:periodicity-test} correctly returned $\Close[Q]$, we have $\HD(Y_v,T)\le \frac{1}{3r_v}$.
    In particular, this implies $\HD(Y_{v,s},T\intervalco{k-s}{|X_v|+k-s})\le \frac{1}{3r_v}$ for every $s\in \intervalcc{-k}{k}$.
    Similarly, since the Matching Test in \cref{alg:stoc:line:alignment-test} correctly returned $\Close[s^*]$, we have
    $\HD(X_v, Y_{v,s^*}) \le \frac{1}{3r_v}$.
    Consequently, $\HD(X_v,S) \le \HD(X_v, Y_{v,s^*})+\HD(Y_{v,s^*},S) \le \frac{2}{3r_v}$.
    By the triangle inequality, for every $s\in \intervalcc{-k}{k}$, we have
    \begin{align*}\big|\ED(X_v,Y_{v,s})-\ED(S,T\intervalco{k-s}{|X_v|+k-s})\big| &\le \ED(X_v,S)+\ED(Y_{v,s},T\intervalco{k-s}{|X_v|+k-s})\\
    & \le \HD(X_v,S)+\HD(Y_{v,s},T\intervalco{k-s}{|X_v|+k-s}) \\ &\le \tfrac{2}{3r_v}+\tfrac{1}{3r_v} = \tfrac{1}{r_v}.\end{align*}

    Strings $S$ and $T$ both have a period $|Q|\le p$ and thus also a period $|Q|\cdot \lceil{p/|Q|\rceil}\in \intervalco{p}{2p}$.
    We apply \cref{lem:approxed-periodic-PM} with the latter period,
    and return the approximation of $\ED(S, T\intervalco{k-s}{|X_v|+k-s})$ as $\eta_{v,s}$.

    Since $|T|-|S|=2k$, the running time is $\Oh(k)+p\Delta \cdot (\log p)^{O(\log_\Delta p)}$,
    and each value $\eta_{v,s}$ is an $(\log p)^{O(\log_\Delta p)}$-approximation of $\ED(S, T\intervalco{k-s}{|X_v|+k-s})$,
    that is,
    \[\ED(S, T\intervalco{k-s}{|X_v|+k-s}) \le \eta_{v,s} \le (\log p)^{O(\log_\Delta p)}\cdot \ED(S, T\intervalco{k-s}{|X_v|+k-s}).\]
    Due to $\big|\ED(X_v,Y_{v,s})-\ED(S,T\intervalco{k-s}{|X_v|+k-s})\big|\le \frac{1}{r_v}$,
    we conclude that
    \[
        \ED(X_v, Y_{v,s})-\tfrac{1}{r_v}\leq \eta_{v,s} \leq (\log p)^{O(\log_\Delta(p))} \cdot \left(\ED(X_v, Y_{v,s})+\tfrac{1}{r_v}\right),
    \]
    holds as claimed.
\end{proof}

Now we are ready to prove the main result of this section, which we restate for convenience.

\lemmainstoc*

\paragraph*{Setting Rates and Accuracies}

Recall that $k, K, p$ and $B$ are given parameters.
Let $T$ be a partition tree of degree $\ell$ for $X, Y$.
For now, we keep $\ell \geq 2$ as a variable that will be set later.
Note that $T$ has at most $2n$ nodes in total, and its depth is bounded by $\ceil{\log_\ell n}$.

Let $\gamma = (\log p)^{c \log_\Delta(p)}$ where the constant $c > 0$ is chosen so that the approximation factor given by \cref{lem:fast-shifts-ed} does not exceed $\gamma$.
We specify the rates and multiplicative accuracies for every node $v$ in $T$ in the following way:
\begin{itemize}
    \item Multiplicative accuracy: if $v$ is the root, we set $\alpha_v = 10 \cdot \gamma$.
    Otherwise, if $v$ is a child of $w$, we set $\alpha_w = \alpha_v \cdot (1 - (2\log n)^{-1})$.
    Note that since the depth of the tree is bounded by $\log n$, for every node $v$ in $T$ it holds that $\alpha_v \geq 10 \gamma = \alpha_{\mathrm{root}}$.
    \item Rate: if $v$ is the root, we set $r_v = 10000 \gamma^2 /  K$.
    Otherwise, if $v$ is a child of $w$ then sample $u_v \sim \calD(\varepsilon:= (2 \log n)^{-1}, \delta := 0.01 \cdot (kn)^{-1})$ (see~\cref{lem:precision-sampling})
    and set $r_v := r_w/u_v$.
\end{itemize}

Recall that our goal is to solve the tree distance problem, that is, for every node $v$ compute values $\eta_{v,s}$ satisfying \eqref{eqn:td-problem}.
We set the parameter $\mu$ in \cref{def:td-problem} to $\mu := \gamma$.

\paragraph*{Correctness} We start by analyzing the correctness of \cref{alg:stoc}.

\begin{lemma}[Correctness of \cref{alg:stoc}]\label{lem:correctness-stoc}
    Let $X, Y$ be strings.
    Given any node $v$ in the partition tree, \cref{alg:stoc} correctly
    solves the Tree Distance Problem (\cref{def:td-problem}), with probability 0.9.
\end{lemma}
\begin{proof}
    We proceed by induction over the depth of the partition tree.
    That is, we will prove that for every node $v$, the algorithm correctly solves $v$ (as per \cref{def:td-problem}) assuming that the recursive calls are correct.

    For the base case, we consider all the nodes solved in \crefrange{alg:stoc:line:test-short-condition}{alg:stoc:line:periodic-solve}.
    If the node is solved in \crefrange{alg:stoc:line:test-short-condition}{alg:stoc:line:test-short}, then we compute the tree distance exactly.
    If the node $v$ is solved in \crefrange{alg:stoc:line:alignment-test-condition}{alg:stoc:line:periodic-solve}, then assuming that the calls to \cref{lem:fast-shifts-ed} succeed (we will bound the error probability later) the algorithm computes values $\eta_{v,s}$ satisfying
    \[
        \ED(X_v, Y_{v,s}) - 1/r_v \leq \eta_{v,s} \leq \gamma \cdot \ED(X_v, Y_{v,s}) + \gamma/r_v.
    \]
    We want to show that these values satisfy \eqref{eqn:td-problem}.
    Observe that the lower bound is satisfied (without additive error).
    For the upper bound, recall that $\alpha_v = 10 \gamma \cdot (1 - (2\log n)^{-1})^d$, where $d \leq \log n$ is the depth of $v$.
    In particular, observe that $\alpha_v \geq \gamma$.
    By \cref{lem:equiv-td-ed}, it holds that $\ED(X_v, Y_{v,s}) \leq \TD^k(X_v, Y_{v,s})$.
    Therefore, we obtain $\eta_{v,s} \leq \alpha_v \cdot \TD^k(X_v, Y_{v,s}) + \gamma/r_v$, as required (since $\mu = \gamma$).

    For the inductive step, fix some node $v$ and assume that all values $\eta_{v_i, s'}$ recursively computed in \cref{alg:stoc:line:recursion} are correct (i.e. they satisfy \eqref{eqn:td-problem}).
    (We will bound the error probability later.)
    In \cref{alg:stoc:line:range-minimum}, the algorithm computes a value $\widetilde{A}_{i,s}$ satisfying

    \begin{equation}
        \widetilde{A}_{i,s}
        = \min_{s' \in \intervalcc{-k}{k}}\eta_{v_i, s'} + 2|s-s'|
        \leq \min_{s' \in \intervalcc{-k}{k}} \alpha_{v_i} \TD^k_{v_i,s'}(X, Y) + 2|s-s'| + \gamma/r_{v_i} \label{proof:correctness:eqn:ub}
    \end{equation}

    \begin{equation}
      \widetilde{A}_{i,s}
        = \min_{s' \in \intervalcc{-k}{k}}\eta_{v_i, s'} + 2|s-s'|
        \geq \min_{s' \in \intervalcc{-k}{k}} 1/\alpha_{v_i} \ED^k_{v_i,s'}(X, Y) + 2|s-s'| - 1/r_{v_i} \label{proof:correctness:eqn:lb}
    \end{equation}
    Recall that for each child $v_i$ the rate $r_{v_i}$ is set to $r_{v_i} = r_v/u_{v_i}$,
    where $u_{v_i}$ is an independent sample from the distribution $\calD(\varepsilon = (2\log n)^{-1}, \delta = 0.01/(kn))$.
    Similarly, recall that the multiplicative accuracy $\alpha_{v_i}$ is set to $\alpha_{v_i} = \alpha_v (1 - (2\log n)^{-1})$ and is the same for all children.

    We prove that the values $\eta_{v,s}$ computed in \cref{alg:stoc:line:recovery} satisfy the upper and lower bounds of \eqref{eqn:td-problem}:

    \paragraph*{Upper Bound}
    Note that by \cref{def:shift-td} it holds that
    \[
      \TD^k_{v,s}(X, Y) = \sum_{i=0}^{\ell-1}A_{i,s}, \quad A_{i,s} = \min_{s' \in \intervalcc{-k}{k}} \TD^k_{v_i, s'}(X, Y) + 2|s-s'|.
    \]
    In particular, combining the above we obtain that $\widetilde{A}_{i, s} \leq \alpha_v (1 - (2\log n)^{-1}) A_{i,s} + \gamma u_{v_i}/r_v$.

    Therefore, we apply \cref{lem:precision-sampling} (with $\alpha = \alpha_v (1 - (2\log n)^{-1})$ and $\beta = \gamma/r_v$) to infer that in \cref{alg:stoc:line:recovery},
    the call to $\alg{Recover}(\widetilde{A}_{0,s},\dots,\widetilde{A}_{\ell-1,s}, u_{v_0},\dots,u_{v_{\ell-1}})$
    computes a value satisfying (again, we assume that the output is correct and bound the error probability later)
    \begin{align*}
        \alg{Recover}(\cdot)
            &\leq (1+\varepsilon)\alpha \left(\sum_{i=0}^{\ell-1}A_{i,s}\right) + \gamma/r_v\\
            &= (1+(2\log n)^{-1}) \alpha_v (1 - (2\log n)^{-1}) \TD_{v,s}^k(X, Y) + \gamma/r_v \\
            &\leq \alpha_v \TD_{v,s}^k(X, Y) + \gamma/r_v,
    \end{align*}
    which concludes the proof for the upper bound.

    \paragraph*{Lower Bound}
    Let $s_0,\dots,s_{\ell-1}$ be the shifts chosen by the algorithm in \cref{alg:stoc:line:range-minimum}.
    Thus, from \eqref{proof:correctness:eqn:lb} we have that $\widetilde{A}_{i,s} \geq 1/\alpha_{v_i} (\ED(X_{v_i}, Y_{v_i,s_i})+2|s-s_i|) - 1/r_{v_i}$.
    Similarly as for the upper bound, suppose that the call to \alg{Recover} in \cref{alg:stoc:line:recovery} succeeds.
    Then, by \cref{lem:precision-sampling} (with $\alpha = \alpha_v (1 - (2\log n)^{-1})$ and $\beta = 1/r_v$)
    the call to $\alg{Recover}(\widetilde{A}_{0,s},\dots,\widetilde{A}_{\ell-1,s}, u_{v_0},\dots,u_{v_{\ell-1}})$
    satisfies
    \begin{align*}
        \alg{Recover}(\cdot)
            &\geq \frac{1}{(1+\varepsilon)\alpha} \left(\sum_{i=0}^{\ell-1}\ED(X_{v_i}, Y_{v_i, s_i}) + 2|s-s_i|\right) - 1/r_{v} \\
            &\geq \frac{1}{(1+\varepsilon)\alpha} \left(\sum_{i=0}^{\ell-1}\ED(X_{v_i}, Y_{v_i, s_i}) + \ED(Y_{v_i,s}, Y_{v_i,s_i})\right) - 1/r_{v} \\
            &\geq \frac{1}{(1+\varepsilon)\alpha} \left(\sum_{i=0}^{\ell-1}\ED(X_{v_i}, Y_{v_i, s})\right) - 1/r_{v} \\
            &\geq \frac{1}{(1+\varepsilon)\alpha} \ED(X_v, Y_{v, s}) - 1/r_{v}
    \end{align*}
    The second inequality holds because we can transform $Y_{v_i,s}$ into $Y_{v_i,s_i}$ by inserting and deleting $|s-s_i|$ characters.
    The third one by the triangle inequality.
    And the last inequality by the subadditivity of edit distance.

    Finally, recall that $\alpha = \alpha_v (1-(2\log n)^{-1})$ and $\varepsilon = (2\log n)^{-1}$. Hence,
    \[
      \frac{1}{(1 + \varepsilon)\alpha} = \frac{1}{(1 + (2\log n)^{-1}) (1 - (2\log n)^{-1}) \alpha_v} \geq \frac{1}{\alpha_v},
    \]
    which concludes the lower bound.

    \paragraph*{Error Probability}
    There are four sources of randomness in the algorithm: the Matching and Periodicity tests in lines \cref{alg:stoc:line:alignment-test,alg:stoc:line:periodicity-test}, the call to \cref{lem:fast-shifts-ed} in \cref{alg:stoc:line:periodic-solve}, and the call to the recovery algorithm of precision sampling (\cref{lem:precision-sampling}) in \cref{alg:stoc:line:recovery}.
    For each node, at least one of the Matching and Periodicity tests fails with probability at most $2\delta = 0.02/n$.
    For every leaf, the call to \cref{lem:fast-shifts-ed} fails with probability $0.01/n$.
    We apply the recovery algorithm of \cref{lem:precision-sampling} with $\delta = 0.01/(kn)$ for $2k$ shifts in every node,
    hence the error probability is bounded by $0.02/n$ as well.
    By a union bound, the error probability for one node is bounded by $0.05/n$.
    Since there are at most $2n$ nodes in the tree, the total error probability is bounded by $0.1$.
\end{proof}

\paragraph*{Running Time Analysis}

We say that a node $v$ in the partition tree is active if the recursive computation of \cref{alg:stoc} reaches $v$.

Recall that we say that a node $v$ in the partition tree is matched if there is a shift $s^* \in \intervalcc{-k}{k}$ such that $X_v = Y_{v,s^*}$.
If $v$ is not matched, we call it unmatched.
The following lemma shows that if $\ED(X, Y) \leq k$, then there are not many unmatched nodes.
\begin{lemma}[{Number of Unmatched Nodes \cite[Lemma 25]{BringmannCFN22}}]\label{lem:unmatched}
    Assume that $\ED(X, Y) \leq k$. If the partition tree has depth $D$, then there are at most $kD$ nodes which are not matched.
\end{lemma}

\begin{lemma}[Number of Active Nodes]\label{lem:active-nodes}
    Suppose that $\ED(X, Y) \leq k$. If the partition tree has depth $D$, then the number of active nodes is $O(kD\ell + \BP_p(X)D\ell)$
    with probability at least 0.96.
\end{lemma}
\begin{proof}
    Recall that there are at most $2n$ nodes in the partition tree.
    Hence, by a union bound, all Matching Tests in \cref{alg:stoc:line:alignment-test} succeed with probability at least $1 - 0.02 = 0.98$.
    Similarly, all Periodicity Tests in \cref{alg:stoc:line:periodicity-test} succeed with probability at least $0.98$.
    Thus, by a union bound, all Matching Tests and Periodicity Tests succeed with probability at least $0.96$.
    We condition on this event for the rest of the proof.

    There are 3 kinds of active nodes $v$: (i) $v$ and all its ancestors are unmatched, (ii) $v$ is matched and all its ancestors are unmatched, and (iii) some ancestor of $v$ (and thus $v$ itself) is matched.

    By \cref{lem:unmatched}, the number of nodes of type (i) is $O(kD)$.
    A node of type (ii) has a parent of type (i), so there are $O(kD\ell)$ such nodes in total.
    Finally, there can be at most $O(\BP_p(X) D\ell)$ nodes of type (iii) in total.
    To see this, observe that such a node will be active only if the periodicity test failed at its parent node.
    But by definition of block periodicity (\cref{def:block-periodicity}), at most $O(\BP_p(X))$ nodes per level fail the periodicity test.
\end{proof}

\begin{lemma}[Running time of \cref{alg:stoc}]\label{lem:runningtime-stoc}
    Let $X, Y$ be strings with $k \leq p$, $k \leq \BP_p(X) \leq B$ and $\ED(X, Y) \leq k$.
    Let $2 \leq \Delta \leq n$ be a parameter.
    Then, \cref{alg:stoc} runs in time
    \[
        \left(\frac{n}{K} \Delta+ pB \cdot \Delta\right) (\log n)^{O(\log_\Delta(n))}
    \]
    with probability at least 0.9.
\end{lemma}
\begin{proof}
    We will set the degree $\ell$ of the partition tree in terms of $\Delta$ later.
    First we bound the expected running time of a single execution of \cref{alg:stoc}, ignoring the cost of recursive calls.
    \begin{itemize}
        \item \cref{alg:stoc:line:test-short-condition,alg:stoc:line:test-short}:
        If $v$ is a leaf, then $|X_v| = 1$.
        Therefore, computing $\eta_{v,s}$ for each $s$ takes $O(1)$ time.
        Thus, the overall time is $O(k)$.
        \item \cref{alg:stoc:line:alignment-test,alg:stoc:line:periodicity-test}:
        Running the Matching and Periodicity Tests (\cref{lem:alignment-test,lem:periodicity-test}) takes time
        $O(r_v |X_v| \log n + k \log |X_v| + p) \leq O(r_v |X_v| \log n + k \log n + p)$.
        \item \crefrange{alg:stoc:line:alignment-test-condition}{alg:stoc:line:periodic-solve}:
        Running \cref{lem:fast-shifts-ed} takes time \[O(p(\log p)^{O(\log_\Delta(p))}\Delta + k) \leq p(\log p)^{O(\log_\Delta(p))}\Delta.\]
        Here we used the assumption that $k \leq p$.
        \item \crefrange{alg:stoc:line:iter-children}{alg:stoc:line:range-minimum}:
        The loop runs for $\ell$ iterations.
        In each iteration, we sample a precision in \cref{alg:stoc:line:precision} in expected $O(1)$-time, perform a recursive computation which we ignore here, and apply \cref{lem:range-minimum} in time $\Order(k)$. The total time is $\Order(k\ell)$.
        \item \crefrange{alg:stoc:line:iter-output}{alg:stoc:line:recovery}:
        The loop runs for $\Order(k)$ iterations and in each iteration we apply the recovery algorithm from \cref{lem:precision-sampling} with parameters $\varepsilon = \Omega(1/\log n)$ and $\delta \geq 1/\poly(n)$.
        Each execution takes time $\Order(\ell \varepsilon^{-2} \log(\delta^{-1})) = \Order(\ell \polylog(n))$.
        Hence, the total time is $O(k \ell \polylog(n))$.
    \end{itemize}
    Thus, the overall time for one execution is
    \begin{equation}
        O(r_v |X_v| \log n + k \ell \polylog(n) + p(\log p)^{O(\log_\Delta(p))}\Delta). \label{proof:runningtime:eqn:1}
    \end{equation}
    We will simplify this term by plugging in the (expected) rate~$r_v$ for any node $v$.

    Recall that~$r_v = 10000 \gamma^2 \cdot (K \cdot u_{v_1} \dots u_{v_d})^{-1}$ where $v_0, v_1, \dots, v_d = v$ is the root-to-node path leading to $v$ and each $u_i$ is sampled from \makebox{$\mathcal D(\varepsilon = (2 \log n)^{-1}, \delta = 0.01 \cdot (kn)^{-1})$}, independently.
    Using the efficiency property of \cref{lem:precision-sampling} with $N = 200 n$, there exist events $\calE_w$ happening each with probability $1 - 1/N$ such that
    \begin{equation*}
        \Ex(1 / u_w \mid \calE_w) \leq \Order(\varepsilon^{-2} \polylog(N, \delta^{-1}, \varepsilon^{-1})) \leq \polylog(n).
    \end{equation*}
    Taking a union bound over all nodes in $T$ (there are at most $2n$ many), the event $\calE := \bigwedge_w \calE_w$ happens with probability at least $0.99$; we will condition on $\calE$ from now on.
    Under this condition, we have:
    \begin{align*}
        \Ex(r_v \mid \calE)
            &= \frac{10000 \cdot \gamma^2}K \prod_{i=1}^d \Ex(1 / u_{v_i} \mid \calE_{v_i}) \\
            &\leq \frac{\gamma^2 \cdot (\log n)^{\Order(d)}}K
            \leq \frac{\gamma^2 \cdot (\log n)^{\Order(\log_\ell(n))}}K.
    \end{align*}

    Let $c$ be the constant so that in the bound above we have $(\log n)^{O(\log_\ell(n))} \leq (\log n)^{c \log_\ell(n)}$.
    We now set $\ell = \ceil{(\log n)^{c \log_\Delta(n)}}$, so that $(\log n)^{c \log_\ell(n)} \leq \Delta$.
    Therefore, $\Ex(r_v \mid \calE) \leq \gamma^2 \Delta/K$.

    Finally, we can bound the total expected running time (conditioned on $\calE$) summing~\eqref{proof:runningtime:eqn:1} over all active nodes $v$:
    \begin{equation*}
        \sum_v \Order\!\left(|X_v| \cdot \frac{\gamma^2 \cdot \Delta}K + k\ell\polylog(n) + p(\log p)^{O(\log_\Delta(p))}\Delta\right)\!.
    \end{equation*}
    To bound the first term, we use that $\sum_w |X_w| = n$ whenever $w$ ranges over all nodes on a fixed level in the partition tree, and thus $\sum_v |X_v| \leq n \cdot D$ where $v$ ranges over all nodes.
    Thus, recalling that $\gamma = (\log p)^{O(\log_{\Delta}(p))}$ we can bound the first term in the sum by
    \[
        \frac{n}{K} \cdot \Delta \cdot (\log p)^{O(\log_\Delta(p))}.
    \]

    The other terms get multiplied by the number of active nodes.
    By \cref{lem:active-nodes}, the number of active nodes is $O(kD\ell  + \BP_p(X)D\ell) \leq O(BD\ell)$
    with probability at least 0.96 (here we used that $\BP_p(X) \leq B$ and $k \leq B$).
    Conditioned on this, the expected time for the second and third terms is bounded by
    \[
        O\left(k\ell\polylog(n) \cdot BD\ell  + p(\log p)^{O(\log_\Delta(p))}\Delta \cdot BD\ell\right).
    \]
    Using that $k \leq p$, and that $D \leq \log n$, this is bounded by
    \[
        O(p(\log p)^{O(\log_\Delta(p))}\Delta \cdot B \ell^2 \polylog(n)) \leq p(\log n)^{O(\log_\Delta(n))}\Delta \cdot B.
    \]
    In the last step we used that $\ell = (\log n)^{O(\log_\Delta(n))}$ and $p \leq n$.

    Combining the above, we conclude that the overall expected running time is bounded by

    \begin{align*}
        \frac{n}{K} \cdot \Delta \cdot (\log p)^{O(\log_\Delta(p))} + p(\log n)^{O(\log_\Delta(n))}\Delta \cdot B \\
        \leq \left(\frac{n}{K} \cdot \Delta + p B \cdot \Delta\right) \cdot (\log n)^{O(\log_\Delta(n))}.
    \end{align*}

    We conditioned on two events: The event $\calE$ and the event that the number of active nodes is bounded by $O(BD\ell)$ (\cref{lem:active-nodes}).
    Both happen with probability at least $0.96$, thus the total success probability is at least $0.9$.
\end{proof}

\begin{proof}[Proof of~\cref{lem:main-stoc}]
    We run \cref{alg:stoc} with 10 times the time budget \[\left(\frac{n}{K} \Delta + pB \cdot \Delta\right) (\log n)^{O(\log_\Delta(n))}\] given by \cref{lem:runningtime-stoc}.
    If the algorithm exceeds the time budget, we interrupt the computation and return \Far{}.
    By the guarantee of \cref{lem:runningtime-stoc} and Markov's inequality, returning \Far{} in this case is correct with probability at least 0.9.

    If the algorithm terminates, then it computes a value $\eta = \eta_{r, 0}$ where $r$ is the root node of the partition tree.
    By \cref{lem:correctness-stoc}, this value satisfies
    \[
        \frac{\ED(X, Y)}{10 \gamma} - \frac{0.0001 K}{\gamma^2}
            \leq \eta
            \leq 10 \gamma \TD^k(X, Y) + \frac{0.0001 K}{\gamma}.
    \]
    Recall that here, $\gamma = (\log p)^{O(\log_\Delta(p))}$ is the approximation factor of \cref{lem:fast-shifts-ed}.
    If $\ED(X, Y) \leq k$, then by \cref{lem:equiv-td-ed} we obtain that
    \[
        \eta \leq 10 \gamma \cdot 3 D\ell \ED(X, Y) + \frac{0.0001K}{\gamma}
            \leq (\log n)^{O(\log_\Delta(n))} \cdot k + \frac{0.0001K}{\gamma}
    \]
    where we used that the depth $D$ of the tree satisfies $D \leq \log n$, that $\ell = (\log n)^{O(\log_\Delta(n))}$ and that $p \leq n$.
    Conversely, if $\ED(X, Y) > K$ then
    \[
        \eta \geq \frac{\ED(X, Y)}{10 \gamma} - \frac{0.0001K}{\gamma^2}
            > \frac{0.099K}{\gamma}.
    \]
    Thus, to distinguish these two cases we need that
    \[
        \frac{0.099K}{\gamma} > (\log n)^{O(\log_\Delta(n))} \cdot k + \frac{0.0001K}{\gamma}.
    \]
    Using that $p \leq n$, this holds if $K/k > (\log n)^{\Theta(\log_\Delta(n))}$ for a sufficiently large hidden constant.
\end{proof}

\section{Fast Edit Distance Approximation for Periodic Strings}\label{sec:shifts}

In this section, we prove \cref{lem:approxed-periodic-PM}. Our main tool is the following recent result:

\begin{theorem}[Dynamic Approximate ED {\cite[Theorem 7.10]{KMS23}}]\label{thm:dynamic}
    There exists a dynamic algorithm that, initialized with integers $2 \le b \le n$,
    maintains (initially empty) strings $X,Y\in \Sigma^{\le n}$ subject to character edits and,
    after each update, outputs an $\Oh(b \log_b n)$-approximation of $\ED(X,Y)$.
    The amortized expected cost of each update is $b^2\cdot (\log n)^{\Oh(\log_b n)}$,
    and each answer is correct with high probability.
    The probabilistic guarantees hold against an oblivious adversary.
\end{theorem}

First, let us change the parametrization to be consistent with the remainder of this paper:
\begin{corollary}\label{cor:dynamic}
    There exists a dynamic algorithm that, initialized with integers $2 \le \Delta \le n$,
    maintains (initially empty) strings $X,Y\in \Sigma^{\le n}$ subject to character edits and,
    after each update, outputs an $(\log n)^{\Oh(\log_\Delta n)}$-approximation of $\ED(X,Y)$.
    The amortized expected cost of each update is $\Delta \cdot (\log n)^{O(\log_\Delta n)}$,
    and each answer is correct with high probability.
    The probabilistic guarantees hold against an oblivious adversary.
\end{corollary}
\begin{proof}
    Let $b = (\log n)^{c \log_\Delta n}$, where the constant $c$ is chosen so that the update time of \cref{thm:dynamic} does not exceed $b^2\cdot (\log n)^{c\log_b n}$.

    If $b\le n$, we simply use the algorithm of \cref{thm:dynamic}.
    Note that $\Delta = (\log n)^{c \log_b n}$, so the approximation ratio is $\Oh(b \log_b n) =  \Oh((\log n)^{c \log_\Delta n}\cdot  \log n) = (\log n)^{\Oh(\log_\Delta n)}$, whereas the update time can be expressed as $b^2 \cdot (\log n)^{c \log_b n} =  (\log n)^{2c \log_\Delta n}\cdot \Delta = \Delta \cdot (\log n)^{\Oh(\log_\Delta n)}$.

    If $b > n$, then $n \le (\log n)^{\Oh(\log_\Delta n)}$, so it suffices to provide $n$-approximations of $\ED(X,Y)$
    with $\Oh(n)$ time per update. For this, we simply maintain $X$ and $Y$, check whether $X=Y$ for upon each update,
    and, depending on the outcome, report $0$ or $n$ as the $n$-approximation of $\ED(X,Y)$.
\end{proof}

We apply \cref{cor:dynamic} to estimate the distances between a pattern $P$ and substrings of a text $T$.

\begin{lemma}\label{lem:approxed-PM}
    Given strings $P$ and $T$ of lengths $m\le n$, respectively,
    and integer $\Delta\in \intervalcc{2}{n}$, one can compute,
    for all $i\in \intervalcc{0}{n-m}$, multiplicative $(\log n)^{\Oh(\log_\Delta n)}$-approximations of $\ED(P, T\intervalco{i}{i+m})$ in total time $n\Delta \cdot (\log n)^{O(\log_\Delta n)}$ correctly with high probability.
\end{lemma}
\begin{proof}
    We apply \cref{cor:dynamic} with the same parameters. Starting with empty $X$ and $Y$, we use $2m$ edits to set $X:=P$ and $Y:=T\intervalco{0}{m}$.
    Next, we iteratively set $Y:=T\intervalco{i}{i+m}$ for subsequent $i\in \intervalcc{0}{n-m}$ to obtain $(\log n)^{\Oh(\log_\Delta n)}$-approximations of $\ED(P, T\intervalco{i}{i+m})$.
    For this, we note that two edits (one insertion and one deletion) are enough to transform $T\intervalco{i}{i+m}$ to $T\intervalco{i+1}{i+m+1}$.

    Overall, we use \cref{cor:dynamic} for an oblivious sequence of $2n$ edits.
    Consequently, the total expected running time is $n\Delta \cdot (\log n)^{O(\log_\Delta n)}$
    and all the answers are correct with high probability.

    By Markov's inequality, the probability that the algorithm exceeds twice its expected time bound is $1/2$.
    By interrupting the algorithm after this time, and repeating the overall process $O(\log n)$ times, we obtain the lemma statement.
\end{proof}

Our next goal is to improve the running time provided that both strings have a common period.
We start with an auxiliary combinatorial lemma that yields a $3$-approximation of the edit distance between two $p$-periodic strings of the same length.
\begin{lemma}\label{lem:ed-approx-periodic}
    Let $P,Q$ be strings of positive length~$p$, let $n$ be a positive integer,
    and denote $d=\floor{n/p}$ and $r=n \bmod p$.
    Then, the edit distance of strings  $X=P^*\intervalco{0}{n}$ and $Y=Q^*\intervalco{0}{n}$ satisfies
    \[
        \ED(X, Y) \leq \ED(P\intervalco{0}{r}, Q\intervalco{0}{r}) + \min_{s \in \Int} (d \cdot \ED(P, Q^{\circlearrowright s}) + 2|s|) \leq 3\ED(X, Y).
    \]
\end{lemma}
\newcommand{\rot}[2]{#1^{\circlearrowright #2}}

\begin{proof}
    Let us start with the lower bound. Fix $s\in \Int$
    and observe that the triangle inequality and the subadditivity of edit distance yield
    \begin{align*}
        \ED(X,Y) & \le \ED(P^d, Q^d) + \ED(P\intervalco{0}{r}, Q\intervalco{0}{r}) \\
        & \le \ED(P^d, (\rot{Q}{s})^d) + \ED(Q^d, (\rot{Q}{s})^d) +  \ED(P\intervalco{0}{r}, Q\intervalco{0}{r}) \\
        & \le d \cdot \ED(P, \rot{Q}{s}) + 2|s| + \ED(P\intervalco{0}{r}, Q\intervalco{0}{r}).
    \end{align*}

    To prove upper bound, let us partition $Y = \bigodot_{i=0}^{d} Y\intervalco{y_i}{y_{i+1}}$ so that
    \[\ED(X,Y) = \sum_{i=0}^{d-1} \ED(P, Y\intervalco{y_i}{y_{i+1}})+\ED(P\intervalco{0}{r}, Y\intervalco{y_d}{n}).\]
    Fix $j\in \intervalco{0}{d}$ such that $\ED(P, Y\intervalco{y_j}{y_{j+1}})$ does not exceed the average \[\tfrac{1}{d}\left(\ED(X,Y)-\ED(P\intervalco{0}{r}, Y\intervalco{y_d}{n})\right).\]
    Take $s := y_j - pj$ so that $\rot{Q}{s}=Y\intervalco{y_j}{y_j+p}$. Thus, by the triangle inequality,
    \begin{align*}
    \ED(P, \rot{Q}{s}) &\le \ED(P, Y\intervalco{y_j}{y_{j+1}}) + \ED(Y\intervalco{y_j}{y_{j+1}}, Y\intervalco{y_j}{y_j+p})
    \\&= \ED(P, Y\intervalco{y_j}{y_{j+1}}) + |y_{j+1}-y_j - p| \\
    &\le 2 \ED(P, Y\intervalco{y_j}{y_{j+1}}) \\
    &\le \tfrac{2}{d} \left(\ED(X,Y)- 2\ED(P\intervalco{0}{r},Y\intervalco{y_d}{n})\right).
    \end{align*}
    Similarly, $Q\intervalco{0}{r}=Y\intervalco{pd}{n}$ implies
    \begin{align*} \ED(P\intervalco{0}{r}, Q\intervalco{0}{r}) & \le \ED(P\intervalco{0}{r},Y\intervalco{y_d}{n}) + \ED(Y\intervalco{pd}{n},Y\intervalco{y_d}{n}) \\
        & \le \ED(P\intervalco{0}{r},Y\intervalco{y_d}{n})+ |y_d-pd| \\
        &\le 2\ED(P\intervalco{0}{r},Y\intervalco{y_d}{n}).
    \end{align*}
    At the same time,
    \[
        2|s| = 2|y_j-pj|\le \ED(X\intervalco{0}{pj}, Y\intervalco{0}{y_j}) + \ED(X\intervalco{pj}{n}, Y\intervalco{y_j}{n})= \ED(X,Y).\]
    Overall, we have
    $\ED(P\intervalco{0}{r}, Q\intervalco{0}{r}) + d \cdot \ED(P, Q^{\circlearrowright s}) + 2|s| \leq  2\ED(P\intervalco{0}{r},Y\intervalco{y_d}{n}) + d\cdot\frac{2}{d}(\ED(X,Y)- \ED(P\intervalco{0}{r},Y\intervalco{y_d}{n})) + \ED(X,Y) = 3\ED(X, Y)$.
\end{proof}

We are now ready to prove \cref{lem:approxed-periodic-PM}, whose statement is repeated below for readers' convenience.
\lemapproxexperiodic*
\begin{proof}
    Let us first argue that, without loss of generality, we can change the length of $T$ to $m+p-1$ while preserving its string period $T\intervalco{0}{p}$.
    \begin{itemize}
        \item If $n < m+p-1$, then the distances $\ED(P, T\intervalco{i}{i+m})$ for $i\in\intervalcc{0}{n-m}$ are unaffected by the extension. Moreover, extending $T$ does not affect the claimed asymptotic running time because the $\Oh(n-m)$ term is dominated by $\Oh(p)$.
        \item If $n > m+p-1$, then  the distances $\ED(P, T\intervalco{i}{i+m})$ for $i\in\intervalco{0}{p}$ are unaffected by the reduction.
        The remaining distances satisfy $\ED(P, T\intervalco{i}{i+m})=\ED(P, T\intervalco{i- p}{i- p+m})$,
        so the approximation of the latter distance can be returned as the approximation of the former.
        This post-processing costs $\Oh(n-m)$ extra time.
    \end{itemize}
    Thus, we henceforth assume $n=m+p-1$.

    We apply \cref{lem:approxed-PM} twice:
    for $P\intervalco{0}{p}$ and $T\intervalco{0}{2p-1}$,
    as well as for $P\intervalco{0}{r}$ and $T\intervalco{0}{p+r-1}$, where $r = m \bmod p$.
    Suppose that, for each $i\in \intervalco{0}{p}$, this yields values $e_i$ and $f_i$
    that are $(\log p)^{\Oh(\log_\Delta p)}$-approximations of $\ED(P\intervalco{0}{p}, T\intervalco{i}{i+p})$
    and  $\ED(P\intervalco{0}{r}, T\intervalco{i}{i+r})$, respectively.
    For every $i\in \intervalco{0}{p}$, our algorithm returns
    \[g_i := f_i + \min_{j\in \intervalco{0}{p}} (\floor{m/p}\cdot e_j+2\min(|i-j|,p-|i-j|)).\]
    The values $g_i$ can be computed as $u$-to-$w_i$ distances in a weighted graph $G$ consisting of $2p+1$ vertices
    $u,v_0,\ldots,v_{p-1},w_{0},\ldots,w_{p-1}$ and the following edges for each $i\in \intervalco{0}{p}$:
    \begin{itemize}
        \item $u \to v_i$ of length $\floor{m/p}\cdot e_i$;
        \item $v_i \leftrightarrow v_{(i+1)\bmod p}$ of length $2$;
        \item $v_i\to w_i$ of length $f_i$.
    \end{itemize}
    This is because the shortest path of the form $u \to v_j \leadsto v_i \to w_i$
    is of length precisely $\floor{m/p}\cdot e_j + 2\min(|i-j|,p-|i-j|) + f_i$.

    The running time is $p\Delta \cdot (\log p)^{O(\log_\Delta p)}$ for the applications of \cref{lem:approxed-PM}
    plus $\Oh(p\log p)$ for Dijkstra's single-source shortest paths algorithm. The first of these terms dominates.

    The correctness stems from \cref{lem:ed-approx-periodic}, which implies that, for every $i\in \intervalco{0}{p}$,
    the following quantity is a $3$-approximation of $\ED(P, T\intervalco{i}{i+m})$:
    \[\ED(P\intervalco{0}{r},T\intervalco{i}{i+r})+\min_{j\in \intervalco{0}{p}} \left(\floor{m/p} \cdot \ED(P\intervalco{0}{p},T\intervalco{j}{j+p}) + 2\min(|i-j|,p-|i-j|)\right).\]
    If we replace the two edit distances by their $(\log p)^{\Oh(\log_\Delta p)}$-approximations,
    we get a $3\cdot (\log p)^{\Oh(\log_\Delta p)}$-approximation of $\ED(P, T\intervalco{i}{i+m})$.
\end{proof}

\bibliographystyle{halpha}
\bibliography{references}

\newcommand{\etalchar}[1]{$^{#1}$}
\begin{thebibliography}{AHWW16}
\expandafter\ifx\csname url\endcsname\relax
  \def\url#1{\texttt{#1}}\fi
\expandafter\ifx\csname doi\endcsname\relax
  \def\doi#1{\burlalt{doi:#1}{http://dx.doi.org/#1}}\fi
\expandafter\ifx\csname urlprefix\endcsname\relax\def\urlprefix{URL }\fi
\expandafter\ifx\csname href\endcsname\relax
  \def\href#1#2{#2}\fi
\expandafter\ifx\csname burlalt\endcsname\relax
  \def\burlalt#1#2{\href{#2}{#1}}\fi

\bibitem[ABW15]{AbboudBW15}
Amir Abboud, Arturs Backurs, and Virginia~Vassilevska Williams.
\newblock Tight hardness results for {LCS} and other sequence similarity
  measures.
\newblock In Venkatesan Guruswami, editor, {\em 56th Annual {IEEE} Symposium on
  Foundations of Computer Science ({FOCS} 2015)}, pages 59--78. {IEEE} Computer
  Society, 2015.
\newblock \doi{10.1109/FOCS.2015.14}.

\bibitem[AHWW16]{AbboudHWW16}
Amir Abboud, Thomas~Dueholm Hansen, Virginia~Vassilevska Williams, and Ryan
  Williams.
\newblock Simulating branching programs with edit distance and friends, or: {A}
  polylog shaved is a lower bound made.
\newblock In Daniel Wichs and Yishay Mansour, editors, {\em 48th Annual {ACM}
  Symposium on Theory of Computing ({STOC} 2016)}, pages 375--388. {ACM}, 2016.
\newblock \doi{10.1145/2897518.2897653}.

\bibitem[AKO10]{AndoniKO10}
Alexandr Andoni, Robert Krauthgamer, and Krzysztof Onak.
\newblock Polylogarithmic approximation for edit distance and the asymmetric
  query complexity.
\newblock In {\em 51st Annual {IEEE} Symposium on Foundations of Computer
  Science ({FOCS} 2010)}, pages 377--386. {IEEE} Computer Society, 2010.
\newblock \doi{10.1109/FOCS.2010.43}.

\bibitem[AN10]{AndoniN10}
Alexandr Andoni and Huy~L. Nguyen.
\newblock Near-optimal sublinear time algorithms for ulam distance.
\newblock In Moses Charikar, editor, {\em 21st Annual {ACM-SIAM} Symposium on
  Discrete Algorithms ({SODA} 2010)}, pages 76--86. {SIAM}, 2010.
\newblock \doi{10.1137/1.9781611973075.8}.

\bibitem[AN20]{AndoniN20}
Alexandr Andoni and Negev~Shekel Nosatzki.
\newblock Edit distance in near-linear time: it's a constant factor.
\newblock In Sandy Irani, editor, {\em 61st Annual {IEEE} Symposium on
  Foundations of Computer Science ({FOCS} 2020)}, pages 990--1001. {IEEE},
  2020.
\newblock \doi{10.1109/FOCS46700.2020.00096}.

\bibitem[And17]{Andoni17}
Alexandr Andoni.
\newblock High frequency moments via max-stability.
\newblock In {\em {IEEE} International Conference on Acoustics, Speech and
  Signal Processing ({ICASSP} 2017)}, pages 6364--6368. {IEEE}, 2017.
\newblock \doi{10.1109/ICASSP.2017.7953381}.

\bibitem[AO12]{AndoniO12}
Alexandr Andoni and Krzysztof Onak.
\newblock Approximating edit distance in near-linear time.
\newblock {\em {SIAM} J. Comput.}, 41(6):1635--1648, 2012.
\newblock \doi{10.1137/090767182}.

\bibitem[BCFN22a]{BringmannCFN22}
Karl Bringmann, Alejandro Cassis, Nick Fischer, and Vasileios Nakos.
\newblock Almost-optimal sublinear-time edit distance in the low distance
  regime.
\newblock In Stefano Leonardi and Anupam Gupta, editors, {\em 54th Annual {ACM}
  Symposium on Theory of Computing ({STOC} 2022)}, pages 1102--1115. {ACM},
  2022.
\newblock \doi{10.1145/3519935.3519990}.

\bibitem[BCFN22b]{BringmannCFN22b}
Karl Bringmann, Alejandro Cassis, Nick Fischer, and Vasileios Nakos.
\newblock Improved sublinear-time edit distance for preprocessed strings.
\newblock In Mikolaj Bojanczyk, Emanuela Merelli, and David~P. Woodruff,
  editors, {\em 49th International Colloquium on Automata, Languages, and
  Programming ({ICALP} 2022)}, volume 229 of {\em LIPIcs}, pages 32:1--32:20.
  Schloss Dagstuhl - Leibniz-Zentrum f{\"{u}}r Informatik, 2022.
\newblock \doi{10.4230/LIPIcs.ICALP.2022.32}.

\bibitem[BCR20]{BrakensiekCR20}
Joshua Brakensiek, Moses Charikar, and Aviad Rubinstein.
\newblock A simple sublinear algorithm for gap edit distance.
\newblock {\em CoRR}, 2020.
\newblock \urlprefix\url{https://arxiv.org/abs/2007.14368}.

\bibitem[BEG{\etalchar{+}}21]{BoroujeniEGHS21}
Mahdi Boroujeni, Soheil Ehsani, Mohammad Ghodsi, MohammadTaghi Hajiaghayi, and
  Saeed Seddighin.
\newblock Approximating edit distance in truly subquadratic time: Quantum and
  mapreduce.
\newblock {\em J. {ACM}}, 68(3):19:1--19:41, 2021.
\newblock \doi{10.1145/3456807}.

\bibitem[BEK{\etalchar{+}}03]{BatuEKMRRS03}
Tugkan Batu, Funda Erg{\"{u}}n, Joe Kilian, Avner Magen, Sofya Raskhodnikova,
  Ronitt Rubinfeld, and Rahul Sami.
\newblock A sublinear algorithm for weakly approximating edit distance.
\newblock In Lawrence~L. Larmore and Michel~X. Goemans, editors, {\em 35th
  Annual {ACM} Symposium on Theory of Computing ({STOC 2003})}, pages 316--324.
  {ACM}, 2003.
\newblock \doi{10.1145/780542.780590}.

\bibitem[BES06]{BatuES06}
Tugkan Batu, Funda Erg{\"{u}}n, and S{\"{u}}leyman~Cenk Sahinalp.
\newblock Oblivious string embeddings and edit distance approximations.
\newblock In {\em 17th Annual {ACM-SIAM} Symposium on Discrete Algorithms
  ({SODA} 2006)}, pages 792--801. {ACM} Press, 2006.
\newblock \doi{10.1145/1109557.1109644}.

\bibitem[BI18]{BackursI18}
Arturs Backurs and Piotr Indyk.
\newblock Edit distance cannot be computed in strongly subquadratic time
  (unless {SETH} is false).
\newblock {\em {SIAM} J. Comput.}, 47(3):1087--1097, 2018.
\newblock \doi{10.1137/15M1053128}.

\bibitem[BJKK04]{BarYossefJKK04}
Ziv Bar{-}Yossef, T.~S. Jayram, Robert Krauthgamer, and Ravi Kumar.
\newblock Approximating edit distance efficiently.
\newblock In {\em 45th Annual {IEEE} Symposium on Foundations of Computer
  Science ({FOCS} 2004)}, pages 550--559. {IEEE} Computer Society, 2004.
\newblock \doi{10.1109/FOCS.2004.14}.

\bibitem[BK15]{BringmannK15}
Karl Bringmann and Marvin K{\"{u}}nnemann.
\newblock Quadratic conditional lower bounds for string problems and dynamic
  time warping.
\newblock In Venkatesan Guruswami, editor, {\em 56th Annual {IEEE} Symposium on
  Foundations of Computer Science ({FOCS} 2015)}, pages 79--97. {IEEE} Computer
  Society, 2015.
\newblock \doi{10.1109/FOCS.2015.15}.

\bibitem[BR20]{BrakensiekR20}
Joshua Brakensiek and Aviad Rubinstein.
\newblock Constant-factor approximation of near-linear edit distance in
  near-linear time.
\newblock In Konstantin Makarychev, Yury Makarychev, Madhur Tulsiani, Gautam
  Kamath, and Julia Chuzhoy, editors, {\em 52nd Annual {ACM} Symposium on
  Theory of Computing ({STOC} 2020)}, pages 685--698. {ACM}, 2020.
\newblock \doi{10.1145/3357713.3384282}.

\bibitem[CDG{\etalchar{+}}20]{ChakrabortyDGKS20}
Diptarka Chakraborty, Debarati Das, Elazar Goldenberg, Michal Kouck{\'{y}}, and
  Michael~E. Saks.
\newblock Approximating edit distance within constant factor in truly
  sub-quadratic time.
\newblock {\em J. {ACM}}, 67(6):36:1--36:22, 2020.
\newblock \doi{10.1145/3422823}.

\bibitem[FW65]{FineW65}
N.~J. Fine and H.~S. Wilf.
\newblock Uniqueness theorems for periodic functions.
\newblock {\em Proceedings of the American Mathematical Society},
  16(1):109--114, 1965.
\newblock \urlprefix\url{http://www.jstor.org/stable/2034009}.

\bibitem[GKKS22]{GoldenbergKKS22}
Elazar Goldenberg, Tomasz Kociumaka, Robert Krauthgamer, and Barna Saha.
\newblock Gap edit distance via non-adaptive queries: Simple and optimal.
\newblock In {\em 63rd Annual {IEEE} Symposium on Foundations of Computer
  Science ({FOCS} 2022)}, pages 674--685. {IEEE}, 2022.
\newblock \doi{10.1109/FOCS54457.2022.00070}.

\bibitem[GKS19]{GoldenbergKS19}
Elazar Goldenberg, Robert Krauthgamer, and Barna Saha.
\newblock Sublinear algorithms for gap edit distance.
\newblock In David Zuckerman, editor, {\em 60th Annual {IEEE} Symposium on
  Foundations of Computer Science ({FOCS} 2019)}, pages 1101--1120. {IEEE}
  Computer Society, 2019.
\newblock \doi{10.1109/FOCS.2019.00070}.

\bibitem[Gra16]{Grabowski16}
Szymon Grabowski.
\newblock New tabulation and sparse dynamic programming based techniques for
  sequence similarity problems.
\newblock {\em Discret. Appl. Math.}, 212:96--103, 2016.
\newblock \doi{10.1016/j.dam.2015.10.040}.

\bibitem[GRS20]{GoldenbergRS20}
Elazar Goldenberg, Aviad Rubinstein, and Barna Saha.
\newblock Does preprocessing help in fast sequence comparisons?
\newblock In Konstantin Makarychev, Yury Makarychev, Madhur Tulsiani, Gautam
  Kamath, and Julia Chuzhoy, editors, {\em 52nd Annual {ACM} Symposium on
  Theory of Computing ({STOC} 2020)}, pages 657--670. {ACM}, 2020.
\newblock \doi{10.1145/3357713.3384300}.

\bibitem[IW05]{IndykW05}
Piotr Indyk and David~P. Woodruff.
\newblock Optimal approximations of the frequency moments of data streams.
\newblock In Harold~N. Gabow and Ronald Fagin, editors, {\em 37th Annual {ACM}
  Symposium on Theory of Computing ({STOC 2005})}, pages 202--208. {ACM}, 2005.
\newblock \doi{10.1145/1060590.1060621}.

\bibitem[KMP77]{KnuthMP77}
Donald~E. Knuth, James~H. {Morris Jr.}, and Vaughan~R. Pratt.
\newblock Fast pattern matching in strings.
\newblock {\em {SIAM} J. Comput.}, 6(2):323--350, 1977.
\newblock \doi{10.1137/0206024}.

\bibitem[KMS23]{KMS23}
Tomasz Kociumaka, Anish Mukherjee, and Barna Saha.
\newblock Approximating edit distance in the fully dynamic model.
\newblock In {\em 64th Annual {IEEE} Symposium on Foundations of Computer
  Science ({FOCS} 2022)}, 2023.

\bibitem[KS20a]{KociumakaS20}
Tomasz Kociumaka and Barna Saha.
\newblock Sublinear-time algorithms for computing {\&} embedding gap edit
  distance.
\newblock In Sandy Irani, editor, {\em 61st Annual {IEEE} Symposium on
  Foundations of Computer Science ({FOCS} 2020)}, pages 1168--1179. {IEEE},
  2020.
\newblock \doi{10.1109/FOCS46700.2020.00112}.

\bibitem[KS20b]{KouckyS20}
Michal Kouck{\'{y}} and Michael~E. Saks.
\newblock Constant factor approximations to edit distance on far input pairs in
  nearly linear time.
\newblock In Konstantin Makarychev, Yury Makarychev, Madhur Tulsiani, Gautam
  Kamath, and Julia Chuzhoy, editors, {\em 52nd Annual {ACM} Symposium on
  Theory of Computing ({STOC} 2020)}, pages 699--712. {ACM}, 2020.
\newblock \doi{10.1145/3357713.3384307}.

\bibitem[Lev66]{Levenshtein66}
Vladimir~I. Levenshtein.
\newblock Binary codes capable of correcting deletions, insertions and
  reversals.
\newblock {\em Soviet Physics Doklady}, 10(8):707--710, 1966.

\bibitem[LMS98]{LandauMS98}
Gad~M. Landau, Eugene~W. Myers, and Jeanette~P. Schmidt.
\newblock Incremental string comparison.
\newblock {\em {SIAM} J. Comput.}, 27(2):557--582, 1998.
\newblock \doi{10.1137/S0097539794264810}.

\bibitem[LV88]{LandauV88}
Gad~M. Landau and Uzi Vishkin.
\newblock Fast string matching with {$k$} differences.
\newblock {\em J. Comput. Syst. Sci.}, 37(1):63--78, 1988.
\newblock \doi{10.1016/0022-0000(88)90045-1}.

\bibitem[MP80]{MasekP80}
William~J. Masek and Mike Paterson.
\newblock A faster algorithm computing string edit distances.
\newblock {\em J. Comput. Syst. Sci.}, 20(1):18--31, 1980.
\newblock \doi{10.1016/0022-0000(80)90002-1}.

\bibitem[NW70]{NeedlemanW70}
Saul~B. Needleman and Christian~D. Wunsch.
\newblock A general method applicable to the search for similarities in the
  amino acid sequence of two proteins.
\newblock {\em Journal of molecular biology}, 48(3):443--453, 1970.
\newblock \doi{10.1016/b978-0-12-131200-8.50031-9}.

\bibitem[OR07]{OstrovskyR07}
Rafail Ostrovsky and Yuval Rabani.
\newblock Low distortion embeddings for edit distance.
\newblock {\em J. {ACM}}, 54(5):23, 2007.
\newblock \doi{10.1145/1284320.1284322}.

\bibitem[Sel74]{Sellers74}
Peter~H. Sellers.
\newblock On the theory and computation of evolutionary distances.
\newblock {\em {SIAM} Journal on Applied Mathematics}, 26(4):787--793, 1974.
\newblock \doi{10.1137/0126070}.

\bibitem[Vin68]{Vintsyuk68}
Taras~K. Vintsyuk.
\newblock Speech discrimination by dynamic programming.
\newblock {\em Cybernetics}, 4(1):52--57, 1968.
\newblock \doi{10.1007/BF01074755}.

\bibitem[WF74]{WagnerF74}
Robert~A. Wagner and Michael~J. Fischer.
\newblock The string-to-string correction problem.
\newblock {\em J. {ACM}}, 21(1):168--173, 1974.
\newblock \doi{10.1145/321796.321811}.

\end{thebibliography}

\end{document}